\documentclass[12pt,draftcls,onecolumn]{IEEEtran}
\ifCLASSINFOpdf
\usepackage[pdftex]{graphicx}
\graphicspath{{../pdf/}{../jpeg/}}
\DeclareGraphicsExtensions{.pdf,.jpeg,.png}
\else
\usepackage[dvips]{graphicx}
\graphicspath{{../eps/}}
\DeclareGraphicsExtensions{.eps}
\fi\usepackage{graphicx}

\usepackage{graphics}
\usepackage{epsfig}
\usepackage{epstopdf}
\usepackage{stfloats}
\usepackage[cmex10]{amsmath}
\usepackage{algorithmic}
\usepackage{array}
\usepackage{mdwmath}
\usepackage{mdwtab}
\usepackage{graphicx}
\usepackage{subfigure}
\usepackage{color}
\usepackage{amsfonts,amssymb}
\usepackage{hyperref} 
\usepackage{multirow}
\usepackage{diagbox}
\usepackage{caption}
\usepackage{bbding}
\usepackage{amsfonts,amsthm,array} 
\usepackage[ruled]{algorithm2e}
\usepackage{makecell}
\usepackage{lineno} 

\usepackage{float}  

\newtheorem{lemma}{Lemma}

\begin{document}
	
	
    \title{Trajectory and Power Design for Aerial Multi-User Covert Communications	\thanks{Manuscript received.}}

    \author{Hongjiang~Lei, 
	Jiacheng~Jiang,
	Imran~Shafique~Ansari, \\
	Gaofeng~Pan, 	
	and~Mohamed-Slim~Alouini 
\thanks{Hongjiang Lei is with the School of Communication and Information Engineering, Chongqing University of Posts and Telecommunications, Chongqing 400065, China, also with Chongqing Key Lab of Mobile Communications Technology, Chongqing 400065, China (e-mail: leihj@cqupt.edu.cn).}
\thanks{Jiacheng~Jiang is with the School of Communication and Information Engineering, Chongqing University of Posts and Telecommunications, Chongqing 400065, China (e-mail: cquptjjc@163.com).}
\thanks{Imran~Shafique~Ansari is with Bharath Institute of Science and Technology, Chennai, India (e-mail: ansarimran@ieee.org).}
\thanks{Gaofeng~Pan is with the School of Cyberspace Science and Technology, Beijing Institute of Technology, Beijing 100081, China (e-mail: gaofeng.pan.cn@ieee.org).}
\thanks{Mohamed-Slim~Alouini is with CEMSE Division, King Abdullah University of Science and Technology (KAUST), Thuwal 23955-6900, Saudi Arabia (e-mail: slim.alouini@kaust.edu.sa).}
}

\maketitle

\begin{abstract}
Unmanned aerial vehicles (UAVs) can provide wireless access to terrestrial users, regardless of geographical constraints, and will be an important part of future communication systems. 
In this paper, a multi-user downlink dual-UAVs enabled covert communication system was investigated, in which a UAV transmits secure information to ground users in the presence of multiple wardens as well as a friendly jammer UAV transmits artificial jamming signals to fight with the wardens. 
The scenario of wardens being outfitted with a single antenna is considered, and the detection error probability (DEP) of wardens with finite observations is researched. 
Then, considering the uncertainty of wardens' location, a robust optimization problem with worst-case covertness constraint is formulated to maximize the average covert rate by jointly optimizing power allocation and trajectory. 
To cope with the optimization problem, an algorithm based on successive convex approximation methods is proposed. 
Thereafter, the results are extended to the case where all the wardens are equipped with multiple antennas. 
After analyzing the DEP in this scenario, a tractable lower bound of the DEP is obtained by utilizing Pinsker's inequality. 
Subsequently, the non-convex optimization problem was established and efficiently coped by utilizing a similar algorithm as in the single-antenna scenario. 
Numerical results indicate the effectiveness of our proposed algorithm.
\end{abstract}

\begin{IEEEkeywords}
	Covert communication, 
	unmanned aerial vehicle, 
	cooperative jamming, 
	trajectory and power optimization, 
\end{IEEEkeywords}
	

\section{Introduction}
\label{sec:introduction}

\subsection{Background and Related Works}
\label{sec:Background and Related Works}

With the advantages of flexible deployment and high mobility, unmanned aerial vehicle (UAV)-assisted wireless communications have been widely used in both civil and military applications, including disaster rescue, surveillance, data gathering, as well as data relaying \cite{ZengY2016Mag}, \cite{KimH2018Mag}. 
For example, the UAV-assisted terrestrial network established an emergency communication framework and expanded wireless coverage in \cite{ZhaoN2019WC}. 
In particular, by designing the position of UAVs or optimizing their trajectory, it is possible to have a high probability that the air-to-ground (A2G) links between UAVs and ground nodes are line-of-sight (LoS) \cite{WuQ2018TWC}. 
Therefore, the trajectory of UAVs has a deterministic impact on the performance of the aerial communication systems and becomes a fundamental problem to be solved in designing airborne communication systems \cite{WuQ2019WC}. 
For instance, in \cite{LiP2020TWC}, the authors studied a UAV-enabled multiple access channel in which multiple ground users transmit individual messages to a mobile UAV in the sky, and the one-dimensional trajectory and resource allocation were designed to improve the transmission rate. 
In \cite{HuaM2020TCOM}, the authors investigated a UAV-aided simultaneous uplink and downlink transmission network, and the three-dimensional (3D) trajectory, transmit power, and the schedule was optimized to improve the throughput.

Physical layer security (PLS) is a prospective technology to improve the security of data delivery in UAV-assisted communication systems. 
Much of the previous work on UAV communications networks focused on preventing classified information from being eavesdropped by adversary \cite{ZhangG2019TWC, WangW2021JSAC, YaoJ2020TWC, DuoB2021ChinaCom}. 
Nevertheless, concealing UAV wireless transmissions to prevent UAV communications from being monitored by adversaries has yet to be addressed, which is a significant challenge in future work of UAV communications \cite{ChenX2023Survey}. 
The most typical example is the application in the military, where once the opponent monitors the transmission behavior of the UAV, they will carry out further actions, which may result in the interruption of communication. 
To ensure the covertness of the communication, it is necessary to make the existence of the confidential wireless transmission undetectable by the adversary to the maximum extent. 
For this reason, covert communication, which ensures low interception rates of transmission, has been studied in \cite{BashBA2015Mag}. 
Existing works indicate that covert communication can be achieved through random noises \cite{WangY2021TVT}, \cite{WangD2022TITS}, interference \cite{ShahzadK2018TWC}, or random transmit power \cite{YanS2019TIFS}. 
A friendly jammer was introduced to assist the communication by actively generating jamming signals to degrade the detection performance of the warden. 
Unlike random noise with uncontrollable behavior, interference can be effectively manipulated by optimizing the transmit power of the jammer. 
For instance, the ground jammer-aided multi-antenna UAV covert communication system was investigated in \cite{DuH2022JSAC} and game theory was utilized to determine the optimal transmit and jamming power allocation scheme. 
In \cite{ChenX2021TVT}, the authors optimized the location and transmit power of the terrestrial multi-antenna jammer to maximize the warden's detection error probability (DEP).
The influence of different monitor strategies of wardens on covertness is becoming a hot research topic. 
For instance, in \cite{ForouzeshM2020TCOM}, the covert rate was maximized by joint optimizing power allocation and relay selection under the covert constraints for the non-colluding and colluding scenarios. 
The authors considered covert communication in the presence of a multi-antenna adversary and analyzed the effect of the number of antennas utilized at the adversary on the achievable throughput of covert communication in \cite{ShahzadK2019TVT}. 
And in \cite{HuJ2020TVT}, a detection strategy based on multiple antennas with beam sweeping to detect the potential transmission of UAV in wireless networks was investigated.

Due to characteristics of low cost, flexible deployment, and LoS wireless links, which provide significant performance improvement, the UAV-aided communication systems have attracted growing research interests \cite{ChenX2022Netw, WangH2019TC, JiangX2021WC, YanS2021JSAC}. 
However, since the broadcast feature of wireless communications and the high probability of LoS links for aerial communications, UAV-enabled communications are more vulnerable to malicious adversaries \cite{CuiM2018TVT}. 
In \cite{JiangX2021WC}, the result that the A2G channel leads to a lower DEP for the warden has been shown. 
Existing studies demonstrated that covert performance can be enhanced by jointly optimizing the UAV's trajectory and transmit power, etc.
The average covert transmission rate was maximized by joint UAV trajectory and transmit power optimization in \cite{ZhouX2019TSP}.
An aerial cooperative covert communication scheme was proposed with finite block length by optimizing the block length and transmit power to maximize the effective transmission bits from the transmitter to the legitimate receiver against a flying warden in \cite{ChenX2021TCOM}.

As spelled out in \cite{ZhouX2019TSP}, the warden's location can be obtained by the cameras or radars mounted on the UAVs. 
But in practical scenarios, since wardens want to hide from the legitimate network, estimating the exact location and channel state information is challenging, especially for wardens with some deception equipment. 
As a result, the assumption of perfect warden location significantly affects covertness performance results. 
Some research investigated covert communication in the scenario of warden location imperfect. 
The authors in \cite{JiangX2021TVT} investigated a multi-user covert communications system with a single ground warden whose locations are assumed to be imperfect. 
In \cite{WangC2023TCOM}, the authors investigated a covert communication scheme assisted by the UAV-intelligent reflecting surface (IRS) to maximize the covert transmission rate in the scenario of warden location imperfect.

\subsection{Motivation and Contributions}
Compared with terrestrial jammers, aerial jammers can dramatically improve the performance of communication systems due to the high mobility and LoS dominating the A2G channel. 
In addition, the combination of UAVs and covert communications allows the advantages of flexibility, high mobility and low cost to exploit. 
It is necessary to consider the existence of multiple detection terrestrial nodes with uncertain locations and equipped with multi-antenna detection equipment. 
Motivated by this, in this work, we consider covert communication with dual collaborative UAV systems with multiple terrestrial location uncertainty wardens to attempt to monitor the transmission of the aerial base stations using a limited number of observations. 
The main contributions of this work are summarized as follows:

\begin{enumerate}
	
	\item We consider a legitimate UAV operates as an aerial base station and transmits messages to confidential users, and the multiple terrestrial location uncertainty wardens try to monitor the transmission. The average covert rate is maximized by jointly optimizing the trajectory of two UAVs and the transmit power of the base station UAV. 
	The scenario of wardens being outfitted with a single antenna is considered, and the DEP of wardens with finite observations is researched. 
	Then, considering the uncertainty of wardens' location, a robust optimization problem with worst-case covertness constraint is formulated to maximize the average covert rate by jointly optimizing power allocation and trajectory. 
	Due to its non-convexity, an successive convex approximation (SCA)-based algorithm is proposed to cope with this intractable non-convex problem which can ensure that the optimization problem converges to a Karush-Kuhn-Tucker (KKT) solution. 
	The simulation results are given to verify the proposed algorithm's efficiency, and the proposed scheme's superiority is well demonstrated.
	
	\item Moreover, the scenarios that all the wardens equipped with multiple antennas are considered 
	The closed-form expression of the DEP is obtained, which is difficult to utilize as the constraint. 
	Based on Pinsker’s inequality, we obtain the lower bound of the DEP, which is utilized as the covert constraint.
	Subsequently, a non-convex optimization problem was established and efficiently coped by utilizing a similar algorithm as in the single-antenna scenario. 
	
	\item Relative to \cite{YanS2021JSAC, ZhouX2019TSP}, wherein an aerial base station was utilized to achieve covert communication, dual-UAVs covert communication systems with multiple terrestrial nodes are considered in this work, which makes the optimization problem more challenging and the DEP more complex.
	
	\item Relative to \cite{JiangX2021TVT, WangC2023TCOM} wherein the presence of a warden equipped with a single antenna is considered in the communication systems, the presence of multiple wardens equipped with multiple antennas is considered in the considered communication systems, which makes handling of DEP more challenging. 
	
	\item 
	Compared to \cite{ChenX2021TVT} wherein the wardens' location is assumed to be completely known at the base station, the place of the wardens in this work is uncertain, which makes the application case of the communication system more reasonable.
\end{enumerate}

\subsection{Organization}
The rest of this paper is organized as follows. 
The system model and problem formulation are provided in Section \ref{sec:SystemModelandProblemFormulation}. 
A SCA-based algorithm is proposed in Section \ref{sec:CovertRateMaximization} to solve the problem. 
Section \ref{sec:AllTheWardensEquippedwithMultipleAntennas} extends the system model to the scenario of all the wardens equipped with multiple antennas. 
In Section \ref{sec:Simulation}, the superiority and effectiveness of our proposed algorithm have been demonstrated via numerical results compared to other benchmark schemes. 
Finally, Section \ref{sec:Conclusion} concludes this paper.

\section{System Model and Problem Formulation}
\label{sec:SystemModelandProblemFormulation}

\begin{figure}[t]
	\centering		
	\includegraphics[width = 3in]{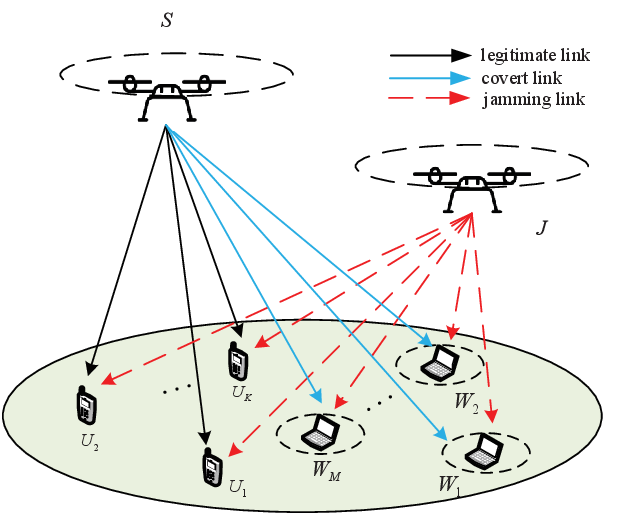}
    \caption{Dual-UAVs enabled covert communication system.}
    \label{fig_model}
\end{figure}

\subsection{System Model}
As shown in Fig. \ref{fig_model}, a flying UAV that works as a base station (${S}$) is transmitting confidential data to ${N_U}$ ground users ($U_k, k =1, \cdots ,N_U$). 
At the same time, $N_W$ terrestrial wardens ($W_m, m = 1, \cdots ,N_W$) monitor the transmission of ${S}$ and ${S}$ tries to hide its transmission from all the wardens. 
A friendly jammer UAV (${J}$) sends jamming signals with the fixed power $P_J$ to fight with the wardens. 
All the wardens are assumed to be equipped with $K$ antennas and the other devices are equipped with a single antenna.

The flight period $T_0$ is equally divided into $N_t$ time slot as ${\delta _t} = \frac{T_0}{N_t}$, where $N_t$ is large enough to guarantee that the position of $S$ and $J$ can be seen as approximately fixed during each time slot \cite{WuQ2018TWC}.
The locations of all nodes in the system model are described using Cartesian coordinate system. 
It is assumed that ${S}$ and ${J}$ fly at a fixed altitude ${H_S}$ and ${H_J}$, respectively.
The horizontal coordinate of $U_k$ and $W_m$ are expressed as 
${{\bf{q}}_{{U_k}}} = {\left[{{x_{{U_k}}},{y_{{U_k}}}} \right]^T}$
and ${{\bf{q}}_{{W_m}}} = {\left[ {{x_{{W_m}}},{y_{{W_m}}}} \right]^T}$, respectively. 
The horizontal locations of $S$ and $J$ are denoted as 
${{\mathbf{q}}_S}{\left( n \right)}= {\left[ {{x_S}{\left( n \right)},{y_S}{\left( n \right)}} \right]^T}$ 
and 
${{\mathbf{q}}_J}{\left( n \right)}= {\left[ {{x_J}{\left( n \right)},{y_J}{\left( n \right)}} \right]^T}$.

It assumed that the A2G links are LoS \cite{WangY2021TCCN} and \cite{ZhangR2021TWC}. 
When $W_m$ is equipped with a single antenna $\left( {K = 1} \right)$, 
the channel coefficients from $S$ to the receivers are expressed as
\begin{equation}
	{h_{S{U_k}}}{\left( n \right)}= \sqrt {\frac{{{\rho _0}}}{{{{\left\| {{{\bf{q}}_S}\left( n \right) - {{\bf{q}}_{{U_k}}}} \right\|}^2} + H_S^2}}},
	\label{4}
\end{equation}
and
\begin{equation}
	{h_{S{W_m}}}{\left( n \right)} = \sqrt {\frac{{{\rho _0}}}{{{{\left\| {{{\bf{q}}_S}\left( n \right) - {{\bf{q}}_{{W_m}}}} \right\|}^2} + H_S^2}}},
	\label{5a}
\end{equation}
respectively,
where ${{\rho _0}}$ denotes the channel power gain at the reference distance.
Similarly, the channel coefficients between $J$ and all the receivers at the $n$th slot are expressed as
\begin{equation}
	{h_{J{U_k}}}{\left( n \right)} = \sqrt {\frac{{{\rho _0}}}{{{{\left\| {{{\mathbf{q}}_J}{\left( n \right)}- {{\bf{q}}_{{U_k}}}} \right\|}^2} + H_J^2}}}, \label{6}
\end{equation}
and
\begin{equation}
	{h_{J{W_m}}}{\left( n \right)} = \sqrt {\frac{{{\rho _0}}}{{{{\left\| {{{\mathbf{q}}_J}{\left( n \right)}- {{\bf{q}}_{{W_m}}}} \right\|}^2} + H_J^2}}}, \label{7a}
\end{equation}
respectively.

To take into account the uncertainty of wardens' position in practical situations, the position of $W_m$ at $S$'s point of view is denoted as \cite{JiangX2021TVT}, \cite{WangC2023TCOM}
\begin{equation}
	{{\bf{q}}_{{W_m}}} = {{\bf{\hat q}}_{{W_m}}} + \Delta {{\bf{q}}_{{W_m}}},
	\label{}
\end{equation}
where ${{\bf{\hat q}}_{{W_m}}}$ is the estimated position of $W_m$ and $\Delta {{\bf{q}}_{{W_m}}}$ denotes the estimation error follows
\begin{equation}
	{\left\| {\Delta {{\bf{q}}_{{W_m}}}} \right\|^2} \le r_m^2,
	\label{}
\end{equation}
where ${r_m}$ is the estimation error radius for $W_m$.

\subsection{Detection Performance}

The DEP is given as 
${\xi _m}\left( n \right) = {P_{{\rm{FA}}_m}}{\left( n \right)}+ {P_{{\rm{MD}}_m}}{\left( n \right)}$, where 
${P_{{\rm{FA}}_m}}{\left( n \right)}$ and ${P_{{\rm{MD}}_m}}{\left( n \right)}$ as $W_m$ denote the false alarm probability (FAP) and miss detection probability (MDP) respectively.
To maximize the covert performance of the UAV system, the DEP of all the wardens should satisfy 
\begin{equation}
	{\xi _m}\left( n \right) \ge 1 - \varepsilon ,\forall n,\forall m, 
	\label{DEP}
\end{equation}
where $0 \le \varepsilon  \le 1$ is an sufficiently small positive value to reflect the requirement for covertness in the communication system.

Similar to \cite{JiangX2021TVT}, assuming that all wardens can only have a finite number of observations, they depend on $I$ times sensing of the received signal to determine whether $S$ is transmitting or not. 
The $i$th $\left( {i = 1, \cdots ,I} \right)$ received signal on $W_m$ is denoted as 
	\begin{equation}
	y_{{W_m}}^i\left( n \right) = \left\{ {\begin{array}{*{20}{c}}
			{\sqrt {{P_J}} {h_{J{W_m}}}\left( n \right)x_J^i\left( n \right) + n_{{W_m}}^i\left( n \right)}&{{{\cal H}_0},}\\
			{\sqrt {{P_S}\left( n \right)} {h_{S{W_m}}}\left( n \right)x_S^i\left( n \right) + \sqrt {{P_J}} {h_{J{W_m}}}\left( n \right)x_J^i\left( n \right) + n_{{W_m}}^i\left( n \right)}&{{{\cal H}_1},}
	\end{array}} \right.
	\label{9}
\end{equation}
where 
${{\mathcal{H}_0}}$ and ${{\mathcal{H}_1}}$ represent the hypothesis that $S$ is not transmitting and transmitting, respectively, 
${{P_S}{\left( n \right)}}$ represents $S$ transmit power, 
$x_S^i{\left( n \right)}$ and $x_J^i{\left( n \right)}$ are the $i$th transmitted signal following complex Gaussian distribution with zero and unit variance,
and 
$n_{{W_m}}^i{\left( n \right)}$ is the complex additive Gaussian noise at $W_m$. 
\setcounter{equation}{8} 
The $i$th received signal at the $m$th warden follows \cite{DeGrootMH2011Book}
\begin{equation}
 	y_{{W_m}}^i{\left( n \right)}\sim \left\{ {\begin{array}{*{20}{l}}
 			{CN\left( {0,\sigma _0^2} \right),}&{{\mathcal{H}_0},} \\ 
 			{CN\left( {0,\sigma _1^2} \right),}&{{\mathcal{H}_1},} 
 	\end{array}} \right. \label{10}
\end{equation}
where $\sigma _0^2= {P_J}{\left| {{h_{J{W_m}}}{\left( n \right)}} \right|^2} + {\sigma ^2}$, 
$\sigma _1^2 = {P_S}{\left( n \right)}{\left| {{h_{S{W_m}}}{\left( n \right)}} \right|^2} + {P_J}{\left| {{h_{J{W_m}}}{\left( n \right)}} \right|^2} + {\sigma ^2}$, 
and ${\sigma ^2}$ is the noise power.
Then the observed matrix at $W_m$ is expressed as
\begin{equation}
	{{\mathbf{y}}_{{W_m}}}{\left( n \right)}= {\left[ {y_{{W_m}}^1{\left( n \right)}, \cdots ,y_{{W_m}}^I{\left( n \right)}} \right]^T} \in {\mathbb{C}^I}. \label{11}
\end{equation}

The worst-scenario is considered where $W_m$ adopts the likelihood ratio test for minimizing its DEP, which is expressed as \cite{WangC2022TWC}, 
\begin{equation}
	\frac{{f\left( {{{\mathbf{y}}_{{W_m}}}{\left( n \right)}|{\mathcal{H}_1}} \right)}}{{f\left( {{{\mathbf{y}}_{{W_m}}}{\left( n \right)}|{\mathcal{H}_0}} \right)}}\mathop {\mathop  \gtrless \limits_{{\mathcal{D}_0}} }\limits^{{\mathcal{D}_1}} 1, 
	\label{12}
\end{equation}
where ${\mathcal{D}_0}$ and ${\mathcal{D}_1}$ represent the $W_m$'s decision, 
$f\left( {{{\mathbf{y}}_{{W_m}}}{\left( n \right)}|{\mathcal{H}_0}} \right)$ and $f\left( {{{\mathbf{y}}_{{W_m}}}{\left( n \right)}|{{\mathcal{H}_1}}} \right)$ are the likelihood function under ${\mathcal{H}_0}$ and ${\mathcal{H}_1}$, respectively, which are expressed as
\begin{equation}
	f\left( {{{\mathbf{y}}_{{W_m}}}{\left( n \right)}|{{\mathcal{H}_j}}} \right) = \frac{1}{{{{\left( {\pi \sigma _j^2} \right)}^I}}}\prod\limits_{i = 1}^I {\exp \left( { - \frac{{{{\left| {y_{{W_m}}^i{\left( n \right)}} \right|}^2}}}{{\sigma _j^2}}} \right)},
	\label{13}
\end{equation}
where 
$j \in \left\{ {0,1} \right\}$.
Following ({\ref{12}}) and ({\ref{13}}), the log-likelihood ratio is obtained as
\begin{equation}
	\begin{aligned}
		{\ell _m}{\left( n \right)}&= \ln \left( {\frac{{f\left( {{{\mathbf{y}}_{{W_m}}}{\left( n \right)}|{{\mathcal{H}_1}}} \right)}}{{f\left( {{{\mathbf{y}}_{{W_m}}}{\left( n \right)}|{{\mathcal{H}_0}}} \right)}}} \right) \\
		&= I\ln \left( {\frac{{\sigma _0^2}}{{\sigma _1^2}}} \right) + \left( {\frac{1}{{\sigma _0^2}} - \frac{1}{{\sigma _1^2}}} \right)\sum\limits_{i = 1}^I {{{\left| {y_{{W_m}}^i{\left( n \right)}} \right|}^2}},
		\label{15}
	\end{aligned}
\end{equation}
and the optimal decision rule at $W_m$ is given by
\begin{equation}
	{\ell _m}{\left( n \right)}\mathop {\mathop  \gtrless \limits_{{\mathcal{D}_0}} }\limits^{{\mathcal{D}_1}} 0  \Leftrightarrow  {E_m}\left( n \right) \mathop {\mathop  \gtrless \limits_{{\mathcal{D}_0}} }\limits^{{\mathcal{D}_1}}  \zeta,
	\label{optimaldecisionrule}
\end{equation}
respectively, 
where 
${E_m}\left( n \right) = \sum\limits_{i = 1}^I {{{\left| {y_{{W_m}}^i\left( n \right)} \right|}^2}} $ 
signify the sum power of each received symbol at $W_m$ and  
$\zeta = I\ln \left( {\frac{{\sigma _1^2}}{{\sigma _0^2}}} \right){\left( {\frac{1}{{\sigma _0^2}} - \frac{1}{{\sigma _1^2}}} \right)^{ - 1}}$
denotes the dection threshold for ${E_m}\left( n \right)$. 
Based on \cite{ShahzadK2019TVT}, 
${E_m}\left( n \right) \sim \frac{{\sigma _j^2}}{2}\chi$,
where 
$\chi $ represents denotes a chi-squared random variable (RV) with $2I$.
Then, following ({\ref{optimaldecisionrule}}), ${P_{{\rm{FA}}_m}}{\left( n \right)}$ is obtained as
\begin{equation}
	\begin{aligned}
		{P_{{\rm{FA}}_m}}{\left( n \right)}& = \Pr \left\{ {{E_m}\left( n \right) \ge \zeta |{{\cal H}_0}} \right\}\\
		&= \Pr \left\{ {\chi  \ge \frac{{2\zeta }}{{\sigma _0^2}}} \right\}\\
		&= 1 - \frac{1}{{\Gamma \left( I \right)}}\Upsilon \left( {I,\frac{\zeta }{{\sigma _0^2}}} \right),
		\label{20}
	\end{aligned}
\end{equation}
where $\Upsilon  \left( {z,x} \right) = \int_ 0^x {{t^{z - 1}}{e^{ - t}}dt}$ is the lower incomplete Gamma function, as defined by \cite[(8.350.1)]{Gradshteyn2007Book}. 
With the same method, we have
\begin{equation}
		{P_{{\rm{MD}}_m}}{\left( n \right)} = \frac{1}{{\Gamma \left( I \right)}}\Upsilon \left( {I,\frac{\zeta }{{\sigma _1^2}}} \right).
		\label{21}
\end{equation}
Thus, the DEP of $W_m$ is obtained as
\begin{equation}
	{\xi _m}\left( n \right) = 1 + \frac{1}{{\Gamma \left( I \right)}}\left( {\Upsilon \left( {I,\frac{\zeta }{{\sigma _1^2}}} \right) - \Upsilon \left( {I,\frac{\zeta }{{\sigma _0^2}}} \right)} \right).
	\label{DEP1}
\end{equation}

\subsection{Problem Formulation}
Assume that the interference signals emitted by $J$ is a Gaussian pseudo-random sequence and can be perfect deleted from the received signals at all the ground user, like \cite{LvL2019TIFS}, \cite{CaiY2018JSAC, XingH2016TVT}. 
Then the achievable rate of $U_k$ is expressed as
\begin{equation}
	{R_k}{\left( n \right)}= {\log _2}\left( {1 + {\gamma _k}\left( n \right)} \right),
	\label{37}
\end{equation}
where 
${\gamma _k}\left( n \right) = {\rho _S}\left( n \right){\left| {{h_{S{U_k}}}\left( n \right)} \right|^2}$
and 
${\rho _S}\left( n \right) = \frac{{{P_S}\left( n \right)}}{{{\sigma ^2}}}$. 
Then the average achievable covert rate of $U_k$ is expressed as 
\begin{equation}
	{\bar R_k} = \frac{1}{N}\sum\limits_{n = 1}^N {{R_k}\left( n \right)}.
	\label{H232}
\end{equation}

In this work, the minimum average covert rate is maximized with respect to the trajectory and transmit power of $S$ and the trajectory of $J$.
Let ${{\mathbf{P}}_S} = \left\{ {{P_S}{\left( n \right)},\forall n} \right\}$, ${{\mathbf{Q}}_S} = \left\{ {{{\mathbf{q}}_S}{\left( n \right)},\forall n} \right\}$, ${{\mathbf{Q}}_J} = \left\{ {{{\mathbf{q}}_J}{\left( n \right)},\forall n} \right\}$, we have the following optimaiztion problem 
\begin{subequations}
	\begin{align}
		\mathcal{P}_{1}: &\mathop {\max }\limits_{{{\mathbf{P}}_S},{{\mathbf{Q}}_S},{{\mathbf{Q}}_J} } \; \mathop {\min }\limits_{1 \le k \le K} {{\bar R}_k} \label{P1a} \\
		{\mathrm{s.t.}}\; & {\xi _m}\left( n \right) \ge 1 - \varepsilon ,\forall n,\forall m, \label{P1b}\\
		& {P_S}\left( n \right) \le {P_{\max }},\forall n, \label{P1c}\\
		&{{\bf{q}}_S}\left( 1 \right) = {{\bf{q}}_S^{\rm{I}}},{{\bf{q}}_S}\left( N \right) = {{\bf{q}}_S^{\rm{F}}}, \label{P1d}\\
		&{{\bf{q}}_J}\left( 1 \right) = {{\bf{q}}_J^{\rm{I}}},{{\bf{q}}_J}\left( N \right) = {{\bf{q}}_J^{\rm{F}}}, \label{P1e}\\
		&\left\| {{{\bf{q}}_S}\left( {n + 1} \right) - {{\bf{q}}_S}\left( n \right)} \right\| \le {V_{S,\max }}{\delta _t},n = 1, \cdots ,{N_t} - 1, \label{P1f}\\
		&\left\| {{{\bf{q}}_J}\left( {n + 1} \right) - {{\bf{q}}_J}\left( n \right)} \right\| \le {V_{J,\max }}{\delta _t},n = 1, \cdots ,{N_t} - 1, \label{P1g}
	\end{align}
\end{subequations}
where 
(\ref{P1b}) is the covert constraint, 
(\ref{P1c}) is the peak power constraint of S where $P_{max}$ signifies maximum transmit power of $S$, 
(\ref{P1d}) and (\ref{P1e}) denote the constraints on the take-off and landing positions of $S$ and $J$, respectively, where 
${{\bf{q}}_S^{\rm{I}}}$ and ${{\bf{q}}_S^{\rm{F}}}$ are the take-off and landing positions of $S$, respectively, 
and 
${{\bf{q}}_J^{\rm{I}}}$ and ${{\bf{q}}_J^{\rm{F}}}$ are the take-off and landing positions of $J$, respectively, 
(\ref{P1f}) and (\ref{P1g}) depicts the maximum flight distance between adjacent time slots of both UAVs 
where 
${V_{S,\max }}$ and ${V_{J,\max }}$ denote the maximum velocity of $S$ and $J$, respectively.

Problem $\mathcal{P}_{1}$ is 
a multivariate coupled non-convex problem. 
Several factors make $\mathcal{P}_{1}$ challenging to solve. 
First, the objective function ${\bar R_k}$ in (\ref{P1a}) is a non-convex with respect to variable ${{\mathbf{P}}_S}$, ${{\mathbf{Q}}_S}$ and ${{\mathbf{Q}}_J}$.
Second, the covert constraint (\ref{P1b}) is a non-convex constraint and it has highly complicated with respect to all the optimization variables.
Moreover, (\ref{P1b}) contains location estimate for all the wardens, which will change the number of constraints from finite to infinite. 

In the following section, a SCA-based algorithm is proposed for handling the non-convexity of problem $\mathcal{P}_{1}$.

\section{Covert Rate Maximization}
\label{sec:CovertRateMaximization}

\subsection{Tractable Covertness Constraint}

Due to the complexity of constraint (\ref{P1b}) and the coupled variables, it is impossible to perform a direct convexity operation on $\mathcal{P}_{1}$. 
The following Lemma provides an alternative to obtain a simpler equivalent form.

\begin{lemma}
	The DEP of $W_m$ is obtained as monotonically decreasing respect to 
	${{\gamma _{m,1}}\left( n \right)} = \frac{{{P_S}\left( n \right){{\left| {{h_{S{W_m}}}\left( n \right)} \right|}^2}}}{{{P_J}{{\left| {{h_{J{W_m}}}\left( n \right)} \right|}^2} + {\sigma ^2}}}$.
\end{lemma}

\begin{proof}
Based on the definition of $\zeta $ and ${\sigma _j^2}$, we obtain
\begin{equation}
	\frac{\zeta }{{\sigma _0^2}} = \frac{{I\left( {1 + {\gamma _{m,1}}\left( n \right)} \right)\ln \left( {1 + {\gamma _{m,1}}\left( n \right)} \right)}}{{{\gamma _{m,1}}\left( n \right)}}
	\label{}
\end{equation}
and 
\begin{equation}
	\frac{\zeta }{{\sigma _1^2}} = \frac{{I\ln \left( {1 + {\gamma _{m,1}}\left( n \right)} \right)}}{{{\gamma _{m,1}}\left( n \right)}},
	\label{}
\end{equation}
where
${{\gamma _{m,1}}\left( n \right)} = \frac{{{P_S}\left( n \right){{\left| {{h_{S{W_m}}}\left( n \right)} \right|}^2}}}{{{P_J}{{\left| {{h_{J{W_m}}}\left( n \right)} \right|}^2} + {\sigma ^2}}}$ 
denotes the SINR of $W_m$. 
Then based on $\Upsilon  \left( {z,x} \right) = \int_ 0^x {{t^{z - 1}}{e^{ - t}}dt}$, the DEP is rewritten as
\begin{equation}
	{\xi _m}\left( n \right) = 1 - \frac{1}{{\Gamma \left( I \right)}}\int_{\frac{{I\ln \left( {1 + {\gamma _{m,1}}\left( n \right)} \right)}}{{{\gamma _{m,1}}\left( n \right)}}}^{\frac{{I\left( {1 + {\gamma _{m,1}}\left( n \right)} \right)\ln \left( {1 + {\gamma _{m,1}}\left( n \right)} \right)}}{{{\gamma _{m,1}}\left( n \right)}}} {{t^{I - 1}}{e^{ - t}}dt}.
	\label{40}
\end{equation}
One can observe that the 
${\frac{{I\ln \left( 1 + {{{\gamma _{m,1}}\left( n \right)}} \right)}}{{{{\gamma _{m,1}}\left( n \right)}}}}  \le 0$ and ${\frac{{I\left( {1 + {{\gamma _{m,1}}\left( n \right)}} \right)\ln \left( {1 + {{\gamma _{m,1}}\left( n \right)}} \right)}}{{{{\gamma _{m,1}}\left( n \right)}}}}  \ge 0$, 
which signifies 
are monotonically decreasing or increasing respect to ${{{\gamma _{m,1}}\left( n \right)} > 0}$. 
Thus, ${\xi _m}\left( n \right)$ is monotonically decreasing respect to ${{{\gamma _{m,1}}\left( n \right)} }$.
\end{proof}

Since a given covertness requirement must to be satisfied, given in (\ref{DEP}), then we have
\begin{equation}
	{{\gamma _{m,1}}\left( n \right)} \le {\gamma _{\max ,1}},
	\label{snr1Wm}
\end{equation}
where ${\gamma _{\max ,1}} $ is maximum value of ${{\gamma _{m,1}}\left( n \right)}$, which can be obtained by utilizing bisection search method.

\subsection{SCA-Based Joint Optimization Algorithm}

By replace ({\ref{P1b}}) for ({\ref{snr1Wm}}) and introduce a slack variable $\eta$, the problem $\mathcal{P}_{1}$ can be equivalently transformed into a more tractable formulation, namely, 
\begin{subequations}
	\begin{align}
		\mathcal{P}_{1.1}: &\mathop {\max }\limits_{{{\mathbf{P}}_S},{{\mathbf{Q}}_S},{{\mathbf{Q}}_J},\eta } \;\eta \label{P1.1a}\\
		{\mathrm{s.t.}}\; &\eta \le \frac{1}{N}\sum\limits_{n = 1}^N {{R_k}{\left( n \right)}} ,\forall k, \label{P1.1b}\\
		&{{\gamma _{m,1}}\left( n \right)} \le {\gamma _{\max ,1}},\forall n,\forall m, \label{P1.1c}\\
		&(\textrm{\ref{P1c}}) - (\textrm{\ref{P1g}}). \label{P1.1d}
	\end{align}
\end{subequations}
It should be noted that ({\ref{P1.1b}}) and ({\ref{P1.1c}}) are non-convex.

To cope with the non-convexity of  (\ref{P1.1b}), 
the lower bound of ${R_k}\left( n \right) $ is obtained as 
\begin{equation}
	\begin{aligned}
		{R_k}\left( n \right) &\ge {\log _2}\left( {1 + \frac{{{\gamma _0}{P_S}\left( n \right)}}{{{d_k}\left( n \right)}}} \right) \\
		&= {\log _2}\left( {{\gamma _0}{P_S}\left( n \right) + {d_k}\left( n \right)} \right) - {\log _2}\left( {{d_k}\left( n \right)} \right)\\
	 	&\mathop  \ge \limits^{\left( a \right)}  {\log _2}\left( {{\gamma _0}{P_S}\left( n \right) + {d_k}\left( n \right)} \right) - {\log _2}\left( {d_k^l\left( n \right)} \right)  \\
	 	&- \frac{{{d_k}\left( n \right) - d_k^l\left( n \right)}}{{d_k^l\left( n \right)\ln \left( 2 \right)}}\\
		&  \buildrel \Delta \over = R_k^{\rm{L},1}\left( n \right),
		\label{P1.1b2}
	\end{aligned}
\end{equation}
where 
${d_k}\left( n \right)$ is a new slack variable, which must satisfy
\begin{equation}
	{d_k}\left( n \right) \ge {\left\| {{{\bf{q}}_S}\left( n \right) - {{\bf{u}}_k}} \right\|^2} + H_S^2,\forall n,\forall k,
	\label{P1.1b3}
\end{equation}
step $\left( a \right)$ is obtained by SCA, 
$d_k^l\left( n \right)$ is a feasible point derived from the $l$th iteration. 
Then, 
(\ref{P1.1b}) is approximately transformed into
\begin{equation}
	\eta \le \frac{1}{N}\sum\limits_{n = 1}^N {R_k^{\rm{L},1}\left( n \right)} ,\forall k, 
	\label{P1.1b4}
\end{equation}

Similarly, 
a slack variable 
${b_m}\left( n \right)$ 
is introduced for handling the non-convexity of (\ref{P1.1c}), which is reformulated as
\begin{equation}
	\frac{{{P_S}\left( n \right){\rho _0}}}{{{{\left\| {{{\bf{q}}_S}\left( n \right) - {{\bf{q}}_{{W_m}}}} \right\|}^2} + H_S^2}} \le {\gamma _{\max ,1}}{b_m}\left( n \right),\forall n,\forall m,
	\label{P1.1c1}
\end{equation}
where 
${b_m}\left( n \right)$ satisfy
\begin{equation}
	{b_m}\left( n \right) \le \frac{{{P_J}{\rho _0}}}{{{{\left\| {{{\bf{q}}_J}\left( n \right) - {{\bf{q}}_{{W_m}}}} \right\|}^2} + H_J^2}} + {\sigma ^2},\forall n,\forall m.
	\label{P1.1c3}
\end{equation}
By introducing slack variables 
${c_m}\left( n \right)$ 
and 
${v_m}\left( n \right)$ 
the constraint is further reformulated as
\begin{subequations}
	\begin{align}
		\frac{{{\rho _0}{P_S}\left( n \right)}}{{{v_m}\left( n \right)}} &\le {\gamma _{\max ,1}}{b_m}\left( n \right),\forall n,\forall m, \label{eq:1c3a}\\
		{b_m}\left( n \right) &\le \frac{{{\rho _0}{P_J}}}{{{c_m}\left( n \right)}} + {\sigma ^2},\forall n,\forall m, \label{eq:1c3b}\\
		{c_m}\left( n \right) &\ge {\left\| {{{\bf{q}}_J}\left( n \right) - {{\bf{q}}_{{W_m}}}} \right\|^2} + H_J^2,\forall n,\forall m, \label{eq:1c3c}\\
		{v_m}\left( n \right) &\le {\left\| {{{\bf{q}}_S}\left( n \right) - {{\bf{q}}_{{W_m}}}} \right\|^2} + H_S^2,\forall n,\forall m, \label{eq:1c3d}
	\end{align}
\end{subequations}
To handle the mutually coupled variables, (\ref{eq:1c3a}) is written as
\begin{equation}
	\ln \left( {{\rho _0}{P_S}\left( n \right)} \right) \le \ln \left( {{\gamma _{\max ,1}}{b_m}\left( n \right)} \right) + \ln \left( {{v_m}\left( n \right)} \right),\forall n,\forall m.
	\label{eq:1c3a2}
\end{equation}
By utilizing SCA, we obtain 
\begin{equation}
	\begin{aligned}
		\ln \left( {{\rho _0}{P_S}\left( n \right)} \right) &\le \ln \left( {{\rho _0}P_S^l\left( n \right)} \right) + \frac{{{P_S}\left( n \right) - P_S^l\left( n \right)}}{{P_S^l\left( n \right)}}  \buildrel \Delta \over = A\left( n \right),
		\label{eq:1c3a3}
	\end{aligned}
\end{equation}
\begin{equation}
	\begin{aligned}
		\frac{1}{{{c_m}\left( n \right)}} &\ge \frac{1}{{c_m^l\left( n \right)}} - \frac{1}{{{{\left( {c_m^l\left( n \right)} \right)}^2}}}\left( {{c_m}\left( n \right) - c_m^l\left( n \right)} \right) \buildrel \Delta \over = B\left( n \right),
		\label{eq:1c3b2}
	\end{aligned}
\end{equation}
respectively, 
where 
$P_S^l\left( n \right)$ 
and 
$c_m^l\left( n \right)$ are feasible point obtained by the $l$th iteration, respectively, 
Then (\ref{eq:1c3a}) and (\ref{eq:1c3b}) are approximately transformed into
\begin{equation}
	A\left( n \right) \le \ln \left( {{\gamma _{\max ,1}}{b_m}\left( n \right)} \right) + \ln \left( {{v_m}\left( n \right)} \right),\forall n,\forall m,
	\label{eq:1c3a4}
\end{equation}
and 
\begin{equation}
	{b_m}\left( n \right) \le {\rho _0}{P_J}{B_m}\left( n \right) + {\sigma ^2},\forall n,\forall m,
	\label{eq:1c3b3}
\end{equation}
respectively.

It must be noted that constraints (\ref{eq:1c3c}) and (\ref{eq:1c3d}) still intractable due to $\Delta {{\bf{q}}_{{W_m}}}$.
The triangle inequality is utilized to obtain the upper bound, which is expressed as \cite{LeiH2023IoTUAV}
\begin{equation}
	\begin{aligned}
		\left\| {{{\bf{q}}_J}\left( n \right) - {{\bf{q}}_{{W_m}}}} \right\| &= \left\| {{{\bf{q}}_J}\left( n \right) - {{{\bf{\hat q}}}_{{W_m}}} - \Delta {{\bf{q}}_{{W_m}}}} \right\|\\
		&\le \left\| {{{\bf{q}}_J}\left( n \right) - {{{\bf{\hat q}}}_{{W_m}}}} \right\| + \left\| {\Delta {{\bf{q}}_{{W_m}}}} \right\|\\
		&\le \left\| {{{\bf{q}}_J}\left( n \right) - {{{\bf{\hat q}}}_{{W_m}}}} \right\| + {r_m}.
		\label{wmerror}
	\end{aligned}
\end{equation}
Then the constraint ({\ref{eq:1c3c}}) is approximated by
\begin{equation}
	{c_m}{\left( n \right)} \ge {\left( {\left\| {{{\mathbf{q}}_J}{\left( n \right)}- {{{\bf{\hat q}}}_{{W_m}}}} \right\| + {r_m}} \right)^2} + H_J^2,\forall n,\forall m, \label{eq:1c3c2}
\end{equation}
By introducing slack variables 
${{\theta _m}\left( n \right)}$ 
and applying the S-procedure technique, constraint ({\ref{eq:1c3d}}) is equivalently transformed into 
	\begin{subequations}
	\begin{align}
		\left( {\begin{array}{*{20}{c}}
				{\left( {1 + {\theta _m}{\left( n \right)}} \right){\mathbf{I}}}&{ - \left( {{{\mathbf{q}}_S}{\left( n \right)}- {{{\bf{\hat q}}}_{{W_m}}}} \right)} \\ 
				{ - {{\left( {{{\mathbf{q}}_S}{\left( n \right)}- {{{\bf{\hat q}}}_{{W_m}}}} \right)}^T}}&{{{\left\| {{{\mathbf{q}}_S}{\left( n \right)}- {{{\bf{\hat q}}}_{{W_m}}}} \right\|}^2} + {\iota _m}{\left( n \right)}} 
		\end{array}} \right) \succcurlyeq 0, \forall n,\forall m, \label{eq:1c3d1} \\
		{\theta _m}\left( n \right) \ge 0,\forall n,\forall m, \label{eq:1c3d2}
	\end{align}
\end{subequations}
where 
${\mathbf{I}}$ denotes the $2 \times 2$ identity matrix,
${\iota _m}{\left( n \right)}\triangleq H_S^2 - {v_m}{\left( n \right)}- {\theta _m}{\left( n \right)}r_m^2$. 
\setcounter{equation}{39} 
By utilizing SCA, we obtain 
\begin{equation}
	\begin{aligned}
		{\left\| {{{\bf{q}}_S}\left( n \right) - {{{\bf{\hat q}}}_{{W_m}}}} \right\|^2} &\ge {\left\| {{\bf{q}}_S^l\left( n \right) - {{{\bf{\hat q}}}_{{W_m}}}} \right\|^2} \\
		&+ 2{\left( {{\bf{q}}_S^l\left( n \right) - {{{\bf{\hat q}}}_{{W_m}}}} \right)^T}\left( {{{\bf{q}}_S}\left( n \right) - {\bf{q}}_S^l\left( n \right)} \right)\\
		&\buildrel \Delta \over = {C_m}\left( n \right),
		\label{eq:1c3d3}
	\end{aligned}
\end{equation}
where ${\bf{q}}_S^l\left( n \right)$ is a feasible point derived from the $l$th iteration. Then ({\ref{eq:1c3d1}}) is approximated by
\begin{equation}
	\left( {\begin{array}{*{20}{c}}
			{\left( {1 + {\theta _m}\left( n \right)} \right){\mathbf{I}}}&{ - \left( {{{\mathbf{q}}_S}\left( n \right) - {{{\mathbf{\hat q}}}_{{W_m}}}} \right)} \\ 
			{ - {{\left( {{{\mathbf{q}}_S}\left( n \right) - {{{\mathbf{\hat q}}}_{{W_m}}}} \right)}^T}}&{{C_m}\left( n \right) + {\iota _m}\left( n \right)} 
	\end{array}} \right) \succcurlyeq 0,\forall n,\forall m.
	\label{eq:eq:1c3d4}
\end{equation}

Let 
${\mathbf{b}}  = \left\{ {{b_m}\left( n \right),\forall n,\forall m} \right\}$, 
${\mathbf{c}} = \left\{ {{c_m}{\left( n \right)},\forall n,\forall m} \right\}$, 
${\mathbf{d}} = \left\{ {{d_k}{\left( n \right)},\forall n,\forall k} \right\}$, 
${\mathbf{v}} = \left\{ {{v_m}{\left( n \right)},\forall n,\forall m} \right\}$, 
and  
${\mathbf{\Theta  }}  = \left\{ {{\theta _m}{\left( n \right)},\forall n,\forall m} \right\}$, 
problem $\mathcal{P}_{1.1}$ is rewritten as
\begin{subequations}
	\begin{align}
		\mathcal{P}_{1.2}: &\mathop {\max}  \limits_{{{\mathbf{P}}_S},{{\mathbf{Q}}_S},{{\mathbf{Q}}_J},{\mathbf{b}},{\mathbf{c}}, {\mathbf{v}},{\mathbf{d}},{\mathbf{\Theta }},\eta } \; \eta \\		
		{\mathrm{s.t.}}\; & (\textrm{\ref{P1.1d}}), 
		(\textrm{\ref{P1.1b3}}), 
		(\textrm{\ref{P1.1b4}}), 
		(\textrm{\ref{eq:1c3a4}}), 
		(\textrm{\ref{eq:1c3b3}}), 
		(\textrm{\ref{eq:1c3c2}}), 
		(\textrm{\ref{eq:1c3d2}}), 
		(\textrm{\ref{eq:eq:1c3d4}}). 
		\label{P1.2}
	\end{align}
\end{subequations}
Now $\mathcal{P}_{1.2}$ is a convex problem, and it can be efficiently solved by CVX toolbox \cite{BoydS2004Book}.

\begin{algorithm}[!tb]
	\caption{Proposed Algorithm for Non-Convex Problem $\mathcal{P}_{1}$}
	\KwIn{Initialization of feasible points.}
	\While
	{$R\left( {{\mathbf{Q}}_S^l,{\mathbf{Q}}_J^l,{\mathbf{P}}_S^l} \right) - R\left( {{\mathbf{Q}}_S^{l - 1},{\mathbf{Q}}_J^{l - 1},{\mathbf{P}}_S^{l - 1}} \right) \le {\hat \xi}$}
	{   1. Solve $\mathcal{P}_{1.2}$;\\
		2. $l = l + 1$;\\
		3. Calculate the objective value $R\left( {{\mathbf{Q}}_S^l,{\mathbf{Q}}_J^l,{\mathbf{P}}_S^l} \right)$.
	}
	\KwOut{$R\left( {{\mathbf{Q}}_S^l,{\mathbf{Q}}_J^l,{\mathbf{P}}_S^l} \right)$ with ${\mathbf{Q}}_S^* = {\mathbf{Q}}_S^l,\;{\mathbf{Q}}_J^* = {\mathbf{Q}}_J^l,\;{\mathbf{P}}_S^* = {\mathbf{P}}_S^l$.}
\end{algorithm}

To solve $\mathcal{P}_{1}$, we propose an iterative algorithm 
whose details of the entire iteration are summarized in Algorithm 1, where the tolerance of the convergence of the algorithm is denoted as $\hat \xi$.

\subsection{Convergence and Complexity Analysis}

By following the proof of \cite{MarksBR1978OR}, it is straightforward to show that the solution generated by Algorithm 1 converges to a KKT solution of problem $\mathcal{P}_{1}$.

The computational complexity of our proposed Algorithm 1 lies in solving SDP $\mathcal{P}_{1.2}$. 
Based on \cite[Theorem 1]{BiswasP2006ACMTSN}, the computational complexity for finding an ${\hat \varepsilon}$-optimal solution of problem $\mathcal{P}_{1.2}$ is $\mathcal{O}\left( {\sqrt {3{N_t}{N_W} + {A_t}} \left( {{3^3}{N_t}{N_W} + 9{A_t}{N_t}{N_W} + A_t^2} \right)\ln \frac{1}{{\hat \varepsilon}}} \right)$, 
where ${A_t} = {N_U} + 4{N_t}{N_W} + {N_t}{N_U} + 3{N_t} + 6$ is the number of constraints that is not a SDP cone. 
Then, the computational complexity of Algorithm 1 is ${N_{num}}\mathcal{O}\left( {\sqrt {3{N_t}{N_W} + {A_t}} \left( {{3^3}{N_t}{N_W} + 9{A_t}{N_t}{N_W} + A_t^2} \right)\ln \frac{1}{{\hat \varepsilon }}} \right)$, 
where ${N_{num}}$ is the maximal iteration number of Algorithm 1.

\section{All The Wardens Equipped with Multiple Antennas}
\label{sec:AllTheWardensEquippedwithMultipleAntennas}

When all the wardens equipped with multiple antennas $\left( {K > 1} \right)$, the channel coefficients between the transmitters and $W_m$ at the $n$th slot are expressed as
\begin{equation}
	{{\mathbf{h}}_{S{W_m}}}{\left( n \right)}= {\left[ {{h_{S{W_m}}}{\left( n \right)}, \cdots ,{h_{S{W_m}}}{\left( n \right)}} \right]^T} \in {\mathbb{R}^K}, 
	\label{5b}
\end{equation}
and
\begin{equation}
	{{\mathbf{h}}_{J{W_m}}}{\left( n \right)}= {\left[ {{h_{J{W_m}}}{\left( n \right)}, \cdots ,{h_{J{W_m}}}{\left( n \right)}} \right]^T} \in {\mathbb{R}^K}, 
	\label{7b}
\end{equation}
respectively.

The received signal at $W_m$ as \cite{ShahzadK2019TVT}
\begin{equation}
	{{\mathbf{Y}}_{{W_m}}}{\left( n \right)}= \left[ {{\mathbf{y}}_{{W_m}}^1{\left( n \right)}, \cdots ,{\mathbf{y}}_{{W_m}}^I{\left( n \right)}} \right] \in {\mathbb{C}^{K \times I}}, \label{24}
\end{equation}
where 
${\bf{y}}_{{W_m}}^i\left( n \right)$ is given in 
\begin{equation}
	{\bf{y}}_{{W_m}}^i\left( n \right) = \left\{ {\begin{array}{*{20}{c}}
			{\sqrt {{P_J}} {{\bf{h}}_{J{W_m}}}\left( n \right)x_J^i\left( n \right) + {\bf{n}}_{{W_m}}^i\left( n \right),}&{{{\cal H}_0},}\\
			{\sqrt {{P_S}\left( n \right)} {{\bf{h}}_{S{W_m}}}\left( n \right)x_S^i\left( n \right) + \sqrt {{P_J}} {{\bf{h}}_{J{W_m}}}\left( n \right)x_J^i\left( n \right) + {\bf{n}}_{{W_m}}^i\left( n \right),}&{{{\cal H}_1},}
	\end{array}} \right.
	\label{YWm}
\end{equation}
and 
${\mathbf{n}}_{{W_m}}^i{\left( n \right)}= {\left[ {n_{{W_m}}^i{\left( n \right)}, \cdots ,n_{{W_m}}^i{\left( n \right)}} \right]^T} \in {\mathbb{C}^K}$.
\setcounter{equation}{46} 
The $i$th received signal at $W_m$ follows the following distribution
\begin{equation}
	{\mathbf{y}}_{{W_m}}^i{\left( n \right)}\sim \left\{ {\begin{array}{*{20}{l}}
			{CN\left( {{\mathbf{0}},{{\mathbf{K}}_0}{\left( n \right)}} \right)}&{{\mathcal{H}_0},} \\ 
			{CN\left( {{\mathbf{0}},{{\mathbf{K}}_1}{\left( n \right)}} \right)}&{{\mathcal{H}_1},} 
	\end{array}} \right. \label{23}
\end{equation}
where 
${{\mathbf{K}}_0}{\left( n \right)} = {\sigma ^2}{\mathbf{I}} + {P_J}{{\mathbf{h}}_{J{W_m}}}{\left( n \right)}{\mathbf{h}}_{J{W_m}}^H{\left( n \right)}$
and 
${{\mathbf{K}}_1}{\left( n \right)} = {\sigma ^2}{\mathbf{I}} + {P_J}{{\mathbf{h}}_{J{W_m}}}{\left( n \right)}{\mathbf{h}}_{J{W_m}}^H{\left( n \right)}$\\ $+ {P_S}{\left( n \right)}{{\mathbf{h}}_{S{W_m}}}{\left( n \right)}{\mathbf{h}}_{S{W_m}}^H{\left( n \right)}$.
Then the PDF of ${{\mathbf{Y}}_{{W_m}}}{\left( n \right)}$ is expressed as 
	\begin{equation}
	\begin{aligned}
		f\left( {{{\bf{Y}}_{{W_m}}}\left( n \right)|{{\cal H}_j}} \right) &= \prod\limits_{i = 1}^I {\frac{1}{{{\pi ^K}\det \left( {{{\bf{K}}_j}\left( n \right)} \right)}}\exp \left( { - {\bf{y}}_{{W_m}}^i{{\left( n \right)}^H}{\bf{K}}_j^{ - 1}\left( n \right){\bf{y}}_{{W_m}}^i\left( n \right)} \right)} \\
		&= {\left( {\frac{1}{{{\pi ^K}\det \left( {{{\bf{K}}_j}\left( n \right)} \right)}}} \right)^I}\exp \left( { - {\rm{Tr}}\left( {{\bf{K}}_j^{ - 1}\left( n \right){{\bf{Y}}_{{W_m}}}\left( n \right){{\bf{Y}}_{{W_m}}}{{\left( n \right)}^H}} \right)} \right),
		\label{pdfYWm}
	\end{aligned}
\end{equation}
where 
$\det \left( {{{\bf{K}}_1}\left( n \right)} \right) = \left( {{\sigma ^2} + {P_J}{{\left\| {{{\bf{h}}_{J{W_m}}}\left( n \right)} \right\|}^2} + {P_S}\left( n \right){{\left\| {{{\bf{h}}_{S{W_m}}}\left( n \right)} \right\|}^2}} \right){\sigma ^{2\left( {K - 1} \right)}}$, 
$\det \left( {{{\bf{K}}_0}\left( n \right)} \right) = \left( {{\sigma ^2} + {P_J}{{\left\| {{{\bf{h}}_{J{W_m}}}\left( n \right)} \right\|}^2}} \right){\sigma ^{2\left( {K - 1} \right)}}$, 
${{\bf{K}}_1}\left( n \right) = {\sigma ^2}{\bf{I}} + {{\bf{g}}_m}\left( n \right){\bf{g}}_m^H\left( n \right)$, 
${\bf{K}}_1^{ - 1}\left( n \right) = \frac{1}{{{\sigma ^2}}}\left( {{\bf{I}} - \frac{{{{\bf{g}}_m}\left( n \right){\bf{g}}_m^H\left( n \right)}}{{{\sigma ^2} + {{\left\| {{{\bf{g}}_m}\left( n \right)} \right\|}^2}}}} \right)$,
where ${\mathbf{K}}_0^{ - 1}{\left( n \right)}$, ${\mathbf{K}}_1^{ - 1}{\left( n \right)}$ is obtained using Woodbury Matrix Identity for matrix inversion \cite{StrangG1993Book}, 
${{\mathbf{g}}_m}{\left( n \right)} = {\left[ {{G_m}{\left( n \right)}, \cdots ,{G_m}{\left( n \right)}} \right]^T} \in {\mathbb{R}^K}$, 
and 
${G_m}{\left( n \right)}\triangleq \sqrt {{P_J}{{\left| {{h_{J{W_m}}}{\left( n \right)}} \right|}^2} + {P_S}{\left( n \right)}{{\left| {{h_{S{W_m}}}{\left( n \right)}} \right|}^2}}$.

\setcounter{equation}{48} 
Similarly, the log-likelihood ratio is denoted as
\begin{equation}
	\begin{aligned}
		{\ell _m}{\left( n \right)} = & \ln \left( {\frac{{f\left( {{{\mathbf{Y}}_{{W_m}}}{\left( n \right)}|{\mathcal{H}_1}} \right)}}{{f\left( {{{\mathbf{Y}}_{{W_m}}}{\left( n \right)}|{\mathcal{H}_0}} \right)}}} \right) \\
		=& I\ln \left( {\frac{{\det \left( {{{\mathbf{K}}_0}{\left( n \right)}} \right)}}{{\det \left( {{{\mathbf{K}}_1}{\left( n \right)}} \right)}}} \right) + \frac{1}{{{\sigma ^2}}}\frac{{{{\left\| {{\mathbf{g}}_m^H{\left( n \right)}{{\mathbf{Y}}_{{W_m}}}{\left( n \right)}} \right\|}^2}}}{{{\sigma ^2} + {{\left\| {{{\mathbf{g}}_m}{\left( n \right)}} \right\|}^2}}} - \frac{1}{{{\sigma ^2}}}\frac{{{{\left\| {{\mathbf{h}}_{J{W_m}}^H{\left( n \right)}{{\mathbf{Y}}_{{W_m}}}{\left( n \right)}} \right\|}^2}}}{{\frac{{{\sigma ^2}}}{{{P_J}}} + {{\left\| {{{\mathbf{h}}_{J{W_m}}}{\left( n \right)}} \right\|}^2}}}. \label{29}
	\end{aligned}
\end{equation}
According to ({\ref{optimaldecisionrule}}), the optimal decision rule is rewritten as
\begin{equation}
	{{\Omega _m}\left( n \right)} \mathop {\mathop  \gtrless \limits_{{\mathcal{D}_0}} }\limits^{{\mathcal{D}_1}} \lambda
	\label{30}
\end{equation}
where 
${{\Omega _m}\left( n \right)} = \frac{{{{\left\| {{\bf{g}}_m^H\left( n \right){{\bf{Y}}_{{W_m}}}\left( n \right)} \right\|}^2}}}{{{\sigma ^2}\left( {{\sigma ^2} + {{\left\| {{{\bf{g}}_m}\left( n \right)} \right\|}^2}} \right)}} - \frac{{{{\left\| {{\bf{h}}_{J{W_m}}^H\left( n \right){{\bf{Y}}_{{W_m}}}\left( n \right)} \right\|}^2}}}{{{\sigma ^2}\left( {\frac{{{\sigma ^2}}}{{{P_J}}} + {{\left\| {{{\bf{h}}_{J{W_m}}}\left( n \right)} \right\|}^2}} \right)}}$
and 
$\lambda  = I\ln \left( {\frac{{\det \left( {{{\bf{K}}_1}\left( n \right)} \right)}}{{\det \left( {{{\bf{K}}_0}\left( n \right)} \right)}}} \right)$ is the threshold of the detector.

Based on \cite{ShahzadK2019TVT}, under ${\mathcal{H}_j}$, the distributions of 
${\mathbf{h}}_{J{W_m}}^H{\left( n \right)}{{\mathbf{Y}}_{{W_m}}}{\left( n \right)}$ 
and ${\mathbf{g}}_m^H{\left( n \right)}{{\mathbf{Y}}_{{W_m}}}{\left( n \right)}$ 
are given by 
$CN\left( {{\mathbf{0}},\det \left( {{{\mathbf{K}}_j}{\left( n \right)}} \right){{\left\| {{{\mathbf{h}}_{J{W_m}}}{\left( n \right)}} \right\|}^2}{\mathbf{I}}} \right)$ 
and 
$CN\left( {{\mathbf{0}},\det \left( {{{\mathbf{K}}_j}{\left( n \right)}} \right){{\left\| {{{\mathbf{g}}_m}{\left( n \right)}} \right\|}^2}{\mathbf{I}}} \right)$, 
respectively. 
As a result, we have 
\begin{equation}
	{{\Omega _m}\left( n \right)} \sim {\kappa _j}\chi,
	\label{31}
\end{equation}
where 
${\kappa _j} = \frac{{\det \left( {{{\bf{K}}_j}\left( n \right)} \right){{\left\| {{{\bf{g}}_m}\left( n \right)} \right\|}^2}}}{{{\sigma ^2}\left( {{\sigma ^2} + {{\left\| {{{\bf{g}}_m}\left( n \right)} \right\|}^2}} \right)}} - \frac{{\det \left( {{{\bf{K}}_j}\left( n \right)} \right){{\left\| {{{\bf{h}}_{J{W_m}}}\left( n \right)} \right\|}^2}}}{{{\sigma ^2}\left( {\frac{{{\sigma ^2}}}{{{P_J}}} + {{\left\| {{{\bf{h}}_{J{W_m}}}\left( n \right)} \right\|}^2}} \right)}}$.

Then, ${P_{{\rm{FA}}_m}}{\left( n \right)}$ and ${P_{{\rm{MD}}_m}}{\left( n \right)}$ are obtained as
\begin{equation}
	\begin{aligned}
		{P_{{\rm{FA}}_m}}{\left( n \right)}&= \Pr \left( {{\Omega _m}{\left( n \right)} \ge \lambda |{\mathcal{H}_0}} \right)\\
		&= 1 - \frac{{\Upsilon \left( {I,\frac{\lambda }{{{\kappa _0}}}} \right)}}{{\Gamma \left( I \right)}}, 
		\label{33}
	\end{aligned}
\end{equation}
and
\begin{equation}
	\begin{aligned}
		{P_{{\rm{MD}}_m}}{\left( n \right)}&= \Pr \left( {{\Omega _m}{\left( n \right)} \le \lambda |{\mathcal{H}_1}} \right)\\
		&= \frac{{\Upsilon \left( {I,\frac{\lambda }{{{\kappa _1}}}} \right)}}{{\Gamma \left( I \right)}}, 
		\label{34}
	\end{aligned}
\end{equation}
respectively.
Thus, the DEP of $W_m$ is obtained as
\begin{equation}
	{\xi _m}\left( n \right) = 1 + \frac{1}{{\Gamma \left( I \right)}}\left( {\Upsilon \left( {I,\frac{\lambda }{{{\kappa _1}}}} \right) - \Upsilon \left( {I,\frac{\lambda }{{{\kappa _0}}}} \right)} \right).
	\label{DEP2}
\end{equation}

By observation (\textrm{\ref{DEP2}}), we find it is difficult to analysis the DEP of $W_m$ by the method utilized in \textbf{Lemma 1}. 
By utilizing the Pinsker’s inequality, the lower bound of the DEP of $W_m$ is expressed as \cite{Lehmann2006Book}
\begin{equation}
	{\xi _m}\left( n \right) \ge 1 - \sqrt {\frac{1}{2}\mathbb{D}\left( {P_0^I||P_1^I} \right)}, 
	\label{Pinsker}
\end{equation}
where ${P_j^I}$ denotes the probability distribution function of the signal vectors including $I$ symbols received at $W_m$ under hypotheses ${{\mathcal{H}_j}}$
and 
${\mathbb{D}\left( {P_0^I||P_1^I} \right)}$ denotes the Kullback-Leibler (KL) divergence from ${P_0^I}$ to ${P_1^I}$ \cite{ShahzadK2019TVT}. Then the covertness constraint in this case is rewritten as
\begin{equation}
	\mathbb{D}\left( {P_0^I||P_1^I} \right) \le 2{\varepsilon ^2}, 
	\label{43}
\end{equation}

The following Lemma provides the analytical expression for the KL-divergence.

\begin{lemma}
	The KL-divergence is expressed as 
	\begin{equation}
		\mathbb{D}\left( {P_0^I||P_1^I} \right) = I\left( {\ln \left( {1 + {\gamma _{m,2}}\left( n \right)} \right) - \frac{{{\gamma _{m,2}}\left( n \right)}}{{1 + {\gamma _{m,2}}\left( n \right)}}} \right),
		\label{KLdiver}
	\end{equation}
	where 
	${\gamma _{m,2}}\left( n \right) = \frac{{{P_S}\left( n \right){{\left\| {{{\bf{h}}_{S{W_m}}}\left( n \right)} \right\|}^2}}}{{{\sigma ^2} + {P_J}{{\left\| {{{\bf{h}}_{J{W_m}}}\left( n \right)} \right\|}^2}}}$.
\end{lemma}

\begin{proof}
	See Appendix \ref{appendicesA}.
\end{proof}

It is easy to find the KL-divergence given in (\textrm{\ref{KLdiver}}) is monotonically increasing respect to ${\gamma _{m,2}}\left( n \right) $. 
Then, by bisection search method the covertness constraint can be equivalently transformed into
\begin{equation}
	{\gamma _{m,2}}\left( n \right)  \le \gamma _{\max ,2}, 
	\label{snr2Wm}
\end{equation}
where $\gamma _{\max ,2} $ is the maximum value of ${\gamma _{m,2}}\left( n \right) $, which results in the maximum the KL-divergence and the minimum DEP.

At this point, a problem similar to $\mathcal{P}_{1.1}$ can be obtained, namely
\begin{subequations}
	\begin{align}
		\mathcal{P}_{1.3}: &\mathop {\max }\limits_{{{\mathbf{P}}_S},{{\mathbf{Q}}_S},{{\mathbf{Q}}_J},\eta } \;\eta \label{P1.3a}\\
		{\mathrm{s.t.}}\; &\eta \le \frac{1}{N}\sum\limits_{n = 1}^N {{R_k}{\left( n \right)}} ,\forall k, \label{P1.3b}\\
		&{\gamma _{m,2}}\left( n \right)  \le \gamma _{\max ,2},\forall n,\forall m, \label{P1.3c}\\
		&(\textrm{\ref{P1.1d}}). \label{P1.3d}
	\end{align}
\end{subequations}

Due to ${\left\| {{{\mathbf{h}}_{J{W_m}}}{\left( n \right)}} \right\|^2} = K{\left| {{h_{J{W_m}}}{\left( n \right)}} \right|^2}$ and ${\left\| {{{\mathbf{h}}_{S{W_m}}}{\left( n \right)}} \right\|^2} = K{\left| {{h_{S{W_m}}}{\left( n \right)}} \right|^2}$, with the same steps in ({\ref{P1.1b2}}) to ({\ref{eq:eq:1c3d4}}), 
$\mathcal{P}_{1.3}$ in the scenarios with multiple-antenna wardens is approximated as
\begin{subequations}
	\begin{align}
		\mathcal{P}_{1.4}: &\mathop {\max }\limits_{{{\mathbf{P}}_S},{{\mathbf{Q}}_S},{{\mathbf{Q}}_J},{\mathbf{b}},{\mathbf{c}},{\mathbf{v}},{\mathbf{d}},{\mathbf{\Theta  }},\eta } \;\eta \\
		{\mathrm{s.t.}}\; &\ln \left( K \right) + A\left( n \right) \le \ln \left( {{\gamma _{\max ,2}}{b_m}\left( n \right)} \right) \nonumber \\
		&\;\;\;\;\;\;\;  + \ln \left( {{v_m}\left( n \right)} \right),\forall n,\forall m, \\
		&{b_m}\left( n \right) \le K{\rho _0}{P_J}{B_m}\left( n \right) + {\sigma ^2},\forall n,\forall m, \\
		&(\textrm{\ref{P1.1d}}), 
		 (\textrm{\ref{P1.1b3}}), 
		 (\textrm{\ref{P1.1b4}}), 
		 (\textrm{\ref{eq:1c3c2}}), 
		 (\textrm{\ref{eq:1c3d2}}), 
		 (\textrm{\ref{eq:eq:1c3d4}}).
		 \label{P1.3}
	\end{align}
\end{subequations}

One can find, when all the wardens equipped with multiple antennas, $\mathcal{P}_{1}$ also can be solved by algorithm similar to the Algorithm 1.

\begin{table}[t]
		\caption{\emph{List of Simulation Parameters.}}
		\begin{center}
			\begin{tabular}{|c|c|c|c|}
				\hline
				\textbf{Notation}   						& \textbf{Value} &\textbf{Notation}   						& \textbf{Value}\\
				\hline
				${{\mathbf{q}}_S^I}$             		& ${\left[ { -100,100} \right]^T}$ &${{\mathbf{q}}_S^F}$               		& ${\left[ { 700,100} \right]^T}$\\
				\hline
				${{{\mathbf{q}}_{{U_k}}}}$              & ${\left[ { 100,200; 300,300; 500,200} \right]^T}$&${{{{\mathbf{\hat q}}}_{{W_m}}}}$     	& \makecell[c]{${\left[ { 100,0; 300,100; 500,0} \right]^T}$\\${\left[ { 100,100; 300,0; 500,100} \right]^T}$}\\
				\hline
				$r_m$                 			& $15$ m, $30$ m, $15$ m&$H_S$  							& $100$ m\\
				\hline
				$H_J$               			& $70$ m&${\sigma ^2}$			& ${{ -120}}$ dBm\\
				\hline
				${P_{\max }}$			& $0.2$ W&${P_{J}}$					& $0.1$ W\\
				\hline
				${V_{S,\max }}$			& $20$ m/s &${V_{J,\max }}$			& $10$ m/s\\
				\hline
				${\rho _0}$				& ${{ -30}}$ dB& $T _0$						 & $100$ s\\
				\hline
				$\varepsilon $						& 0.05& ${\delta _t}$						& $2$ s\\
				\hline
				$I$										& 30 & $\hat \xi$						& $0.001$\\
				\hline
			\end{tabular}
		\end{center}
		\label{table1}
\end{table}

\section{Simulation Results and Discussion}
\label{sec:Simulation}

In this section, simulation results are presented to verify the performance of the proposed algorithm. 
Unless otherwise stated, the details of the parameter setup are listed in TABLE \ref{table1}. 
For comparison, the following benchmark schemes are also considered: 

\begin{enumerate}
	
	\item Benchmark 1: $S$ flies with constant trajectory and the transmit power of $S$ and the trajectory of $J$ are optimized. 
	
	\item Benchmark 2 \footnote{
		Benchmark 2 is analogous to the scheme proposed in \cite{JiangX2021TVT}, in which the power and trajectory of the aerial base station and user scheduling were jointly optimized to improve the average covert rate of an aerial system. 
		Nevertheless, the airborne jammer was not applied to the communication system in \cite{JiangX2021TVT}. 
		For the purpose of comparison in the same case, we suppose that $J$ flies with a fixed trajectory.
	}: The trajectory of $J$ is fixed and the transmit power and trajectory of $S$ are optimized. 
	
	\item Benchmark 3 \footnote{
		Benchmark 3 is analogous to the optimization scheme proposed in \cite{ChenX2021TVT}, in which a multi-antenna terrestrial jammer was utilized to improve the covert performance of the aerial communication systems and The location of the jammer was optimized. 
		To compare with the existing scheme under the same scenario, we assume $J$ hovering in a fixed position and the hovering position of $J$ was optimized. 
	}: $J$ hovers at a fixed position and the hovering position of $J$, the trajectory and transmit power of $S$ are jointly optimized.
	
\end{enumerate}

\begin{figure}[t]
	\centering
	\subfigure[The optimal trajectory.]{
		\label{fig02a}
		\includegraphics[width = 0.4  \textwidth]{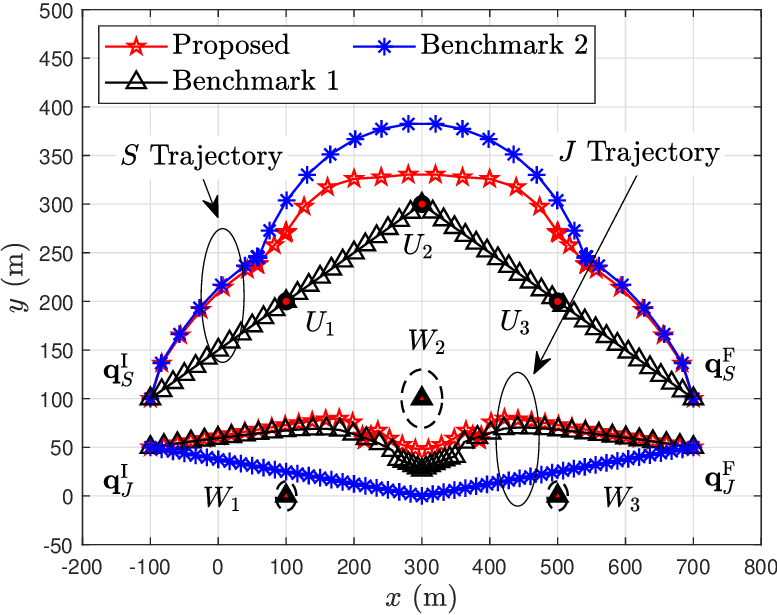}}
	\subfigure[The transmit power of $S$.]{
		\label{fig02b}
		\includegraphics[width = 0.4  \textwidth]{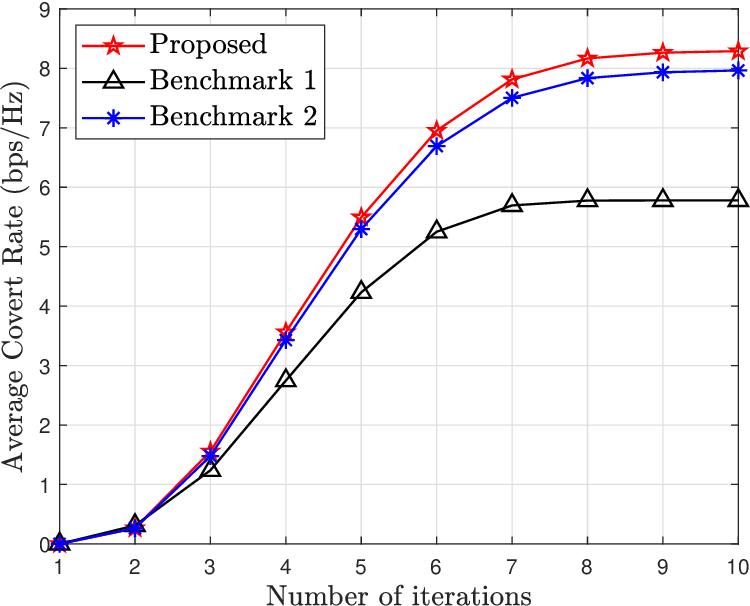}}
	\subfigure[The achievable covert rate.]{
		\label{fig02c}
		\includegraphics[width = 0.4  \textwidth]{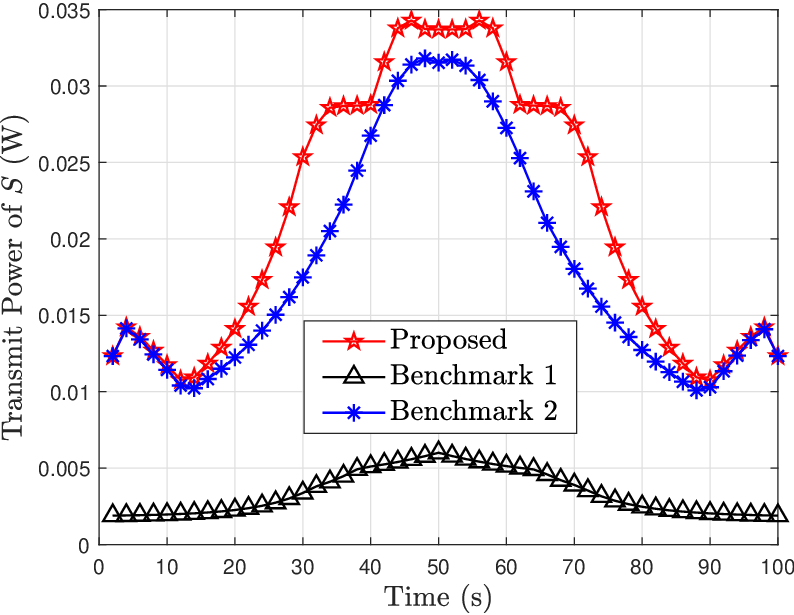}}
	\subfigure[The relationship between the average secrecy rate and the
	number of iterations.]{
		\label{fig02d}
		\includegraphics[width = 0.4  \textwidth]{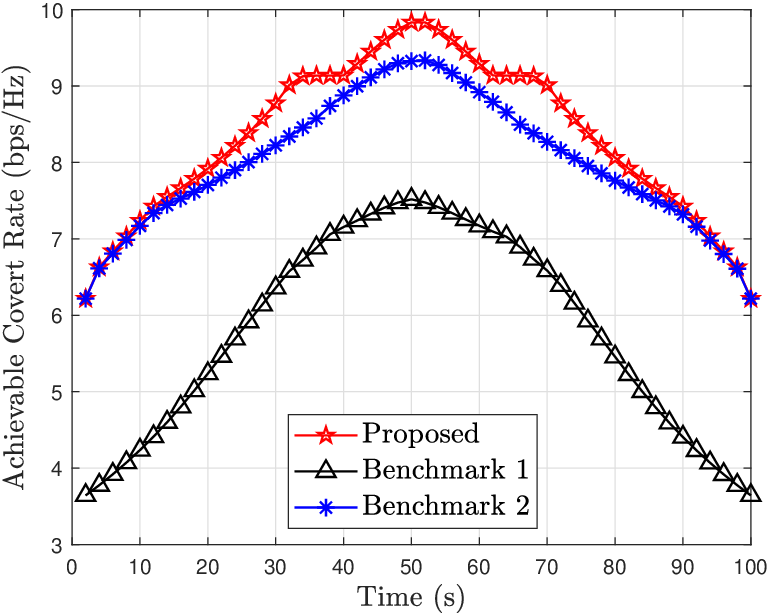}}
	\caption{Scenario 1: The covert systems with $N_W = 3$ single-antenna wardens.}
	\label{fig02s}
\end{figure}
\begin{figure}[t]
	\centering
	\subfigure[The optimal trajectory.]{
		\label{fig03a}
		\includegraphics[width = 0.4  \textwidth]{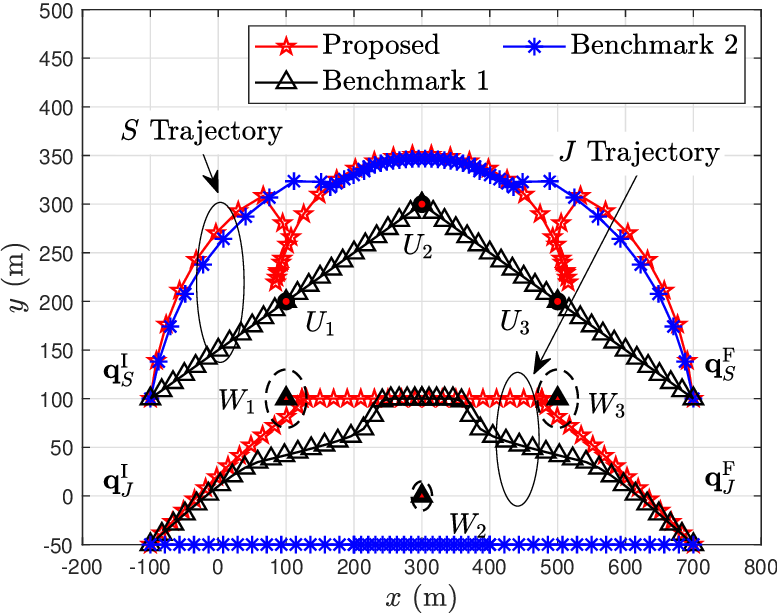}}
	\subfigure[The transmit power of $S$.]{
		\label{fig03b}
		\includegraphics[width = 0.4  \textwidth]{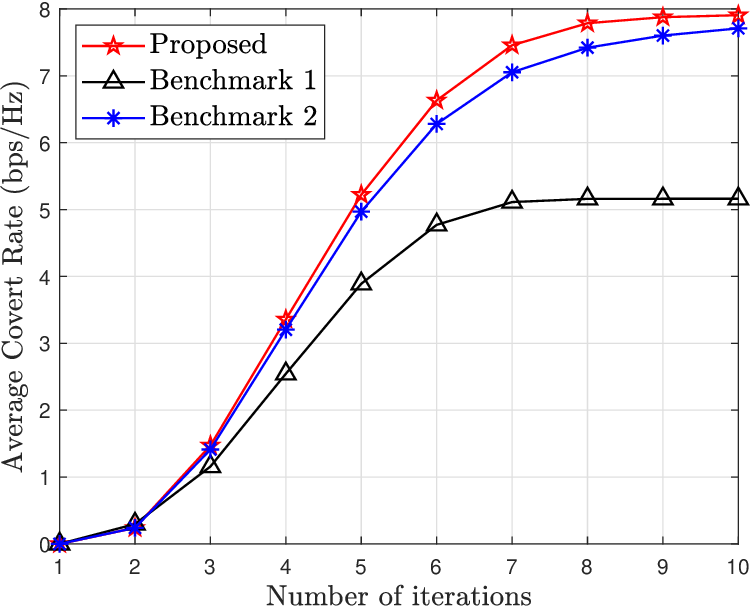}}
	\subfigure[The achievable covert rate.]{
		\label{fig03c}
		\includegraphics[width = 0.4  \textwidth]{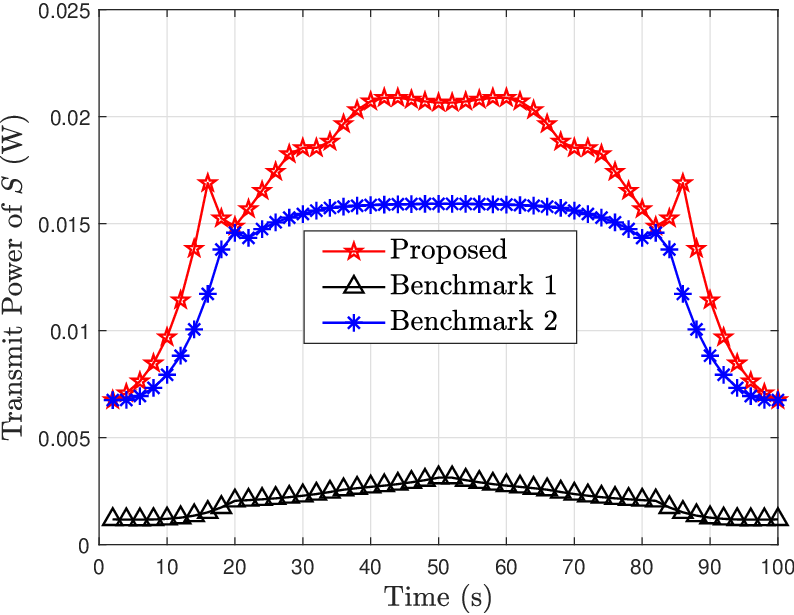}}
	\subfigure[The relationship between the average secrecy rate and the
	number of iterations.]{
		\label{fig03d}
		\includegraphics[width = 0.4  \textwidth]{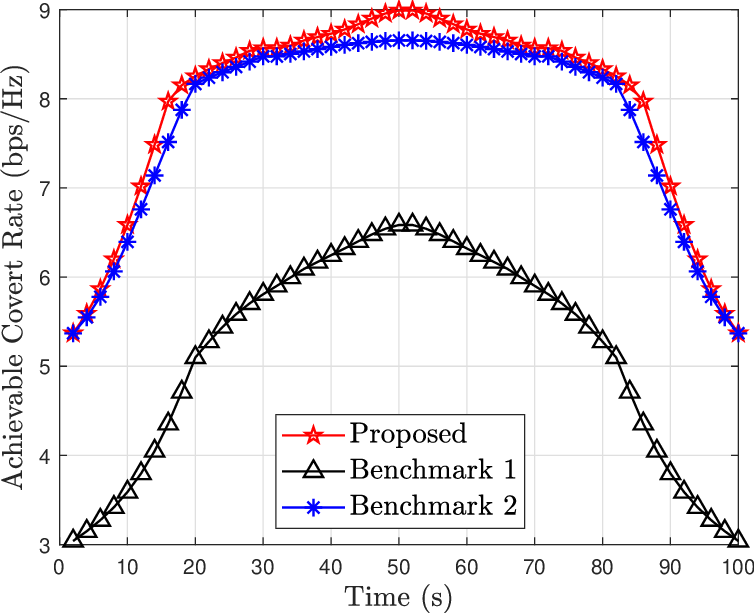}}
	\caption{Scenario 2: The covert systems with $N_W = 3$ single-antenna wardens.}
	\label{fig03s}
\end{figure}
Figs. \ref{fig02s}  and \ref{fig03s} illustrate the optimized trajectory, transmit power and achievable covert rate for the different schemes, respectively. 
For Benchmark 1 where $S$ flies with a given trajectory, the transmit power should be reduced as it approaches the wardens' area to avoid the information being detected. 
While for Benchmark 2 where $J$ flies with a given trajectory, it can be found that when $J$'s trajectory deviates from the optimal trajectory and moves away from the warden that should be close, $S$'s trajectory will also move away from the warden accordingly to prevent the information being detected. 
In the proposed scheme, $S$ is as close as possible to legitimate users and increases the transmit power while satisfying the quality of covert communication to maximize the achievable covert rate. 
At the same time, to make $S$ can be as close as possible to the legitimate user and increase the transmit power, $J$ will get as close as possible to the most threatening warden under the current time slot. 
In the proposed scheme, the characteristics of benchmark 1 and benchmark 2 are considered simultaneously to obtain optimal covert performance.


\begin{figure}[t]
	\centering
	\subfigure[The optimal trajectory.]{
		\label{fig04a}
		\includegraphics[width = 0.4  \textwidth]{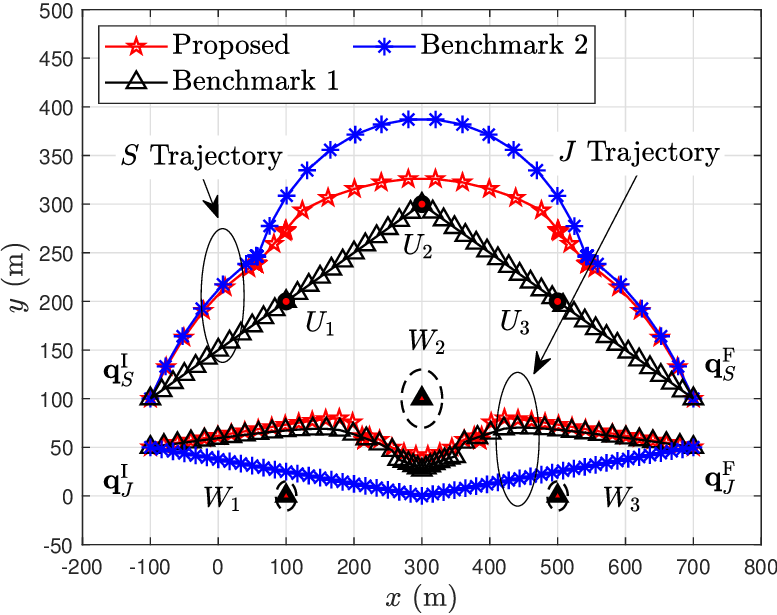}}
	\subfigure[The transmit power of $S$.]{
		\label{fig04b}
		\includegraphics[width = 0.4  \textwidth]{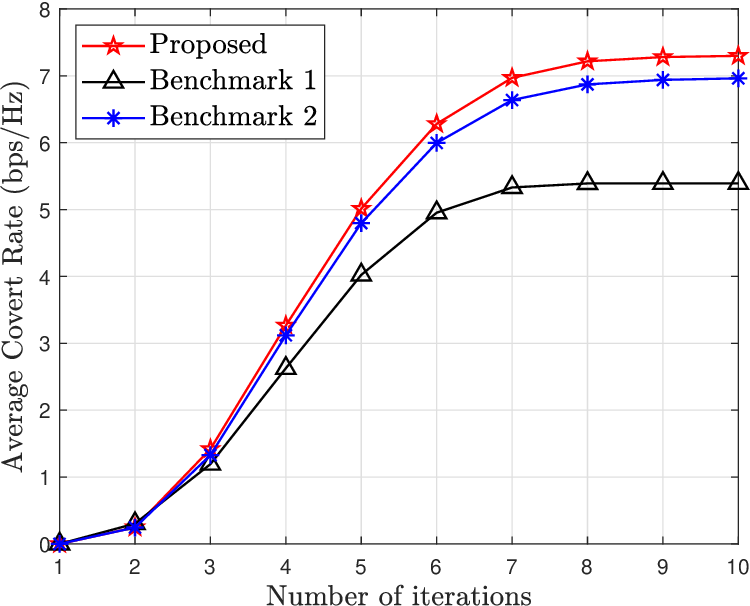}}
	\subfigure[The achievable covert rate.]{
		\label{fig04c}
		\includegraphics[width = 0.4  \textwidth]{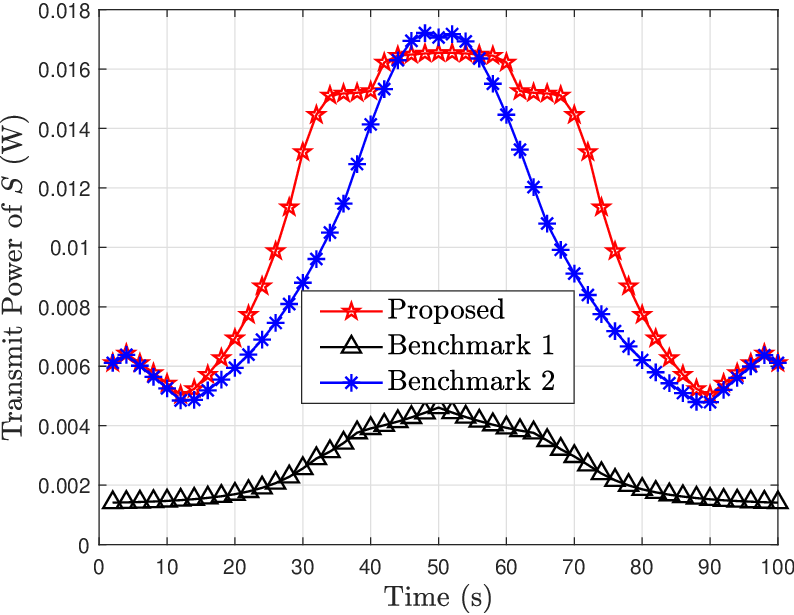}}
	\subfigure[The relationship between the average secrecy rate and the
	number of iterations.]{
		\label{fig04d}
		\includegraphics[width = 0.4  \textwidth]{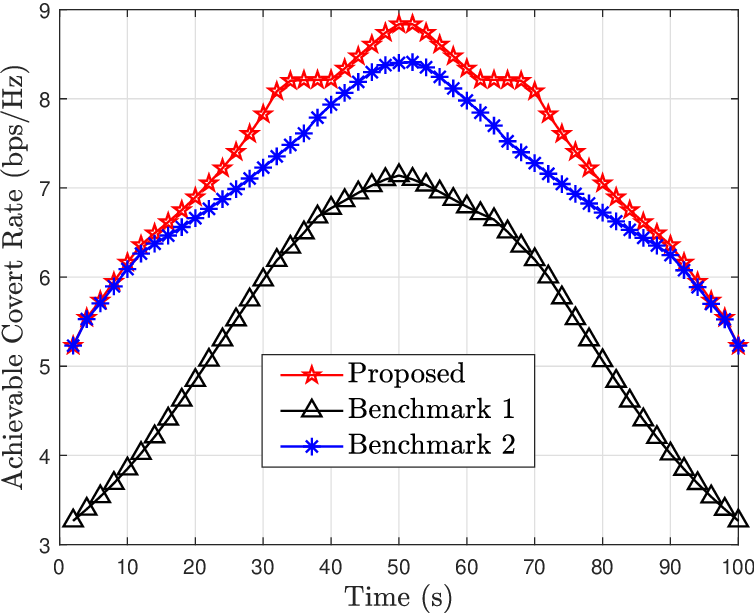}}
	\caption{Scenario 1: The covert systems with $N_W = 3$ wardens with $K =6$ antennas.}
	\label{fig04m}
\end{figure}
\begin{figure}[t]
	\centering
	\subfigure[The optimal trajectory.]{
		\label{fig05a}
		\includegraphics[width = 0.4  \textwidth]{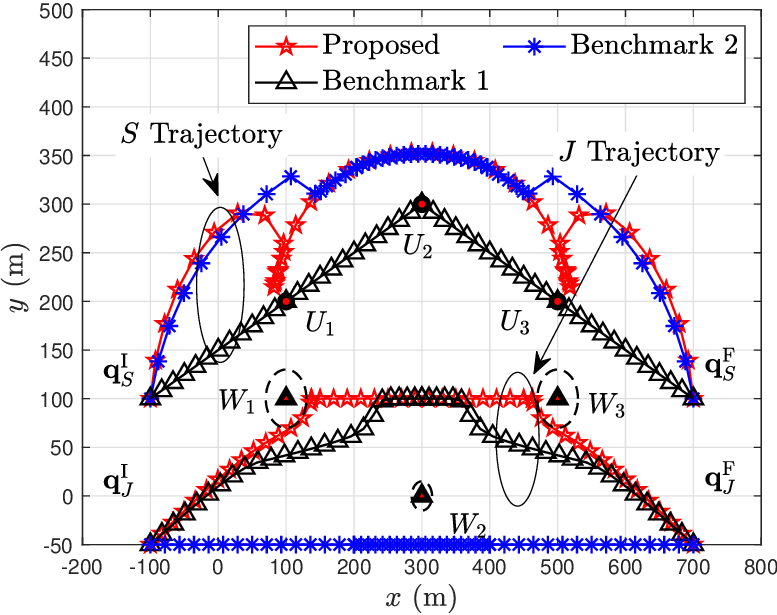}}
	\subfigure[The transmit power of $S$.]{
		\label{fig05b}
		\includegraphics[width = 0.4  \textwidth]{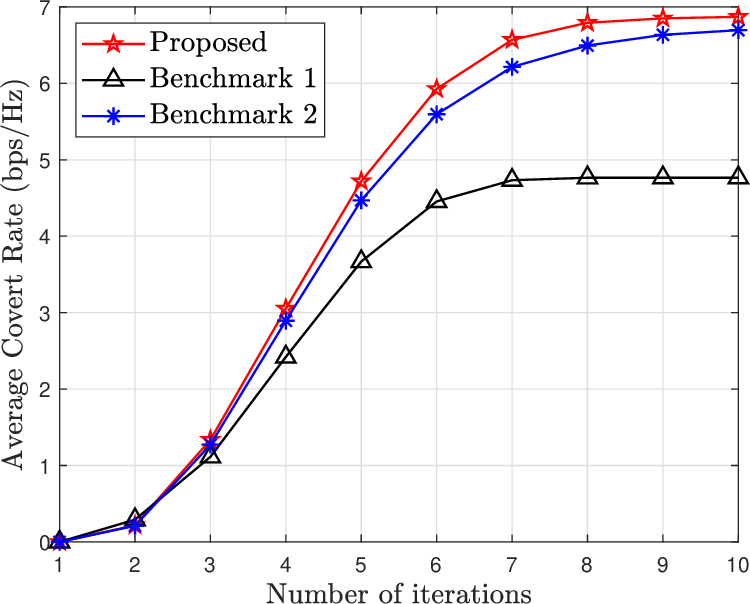}}
	\subfigure[The achievable covert rate.]{
		\label{fig05c}
		\includegraphics[width = 0.4  \textwidth]{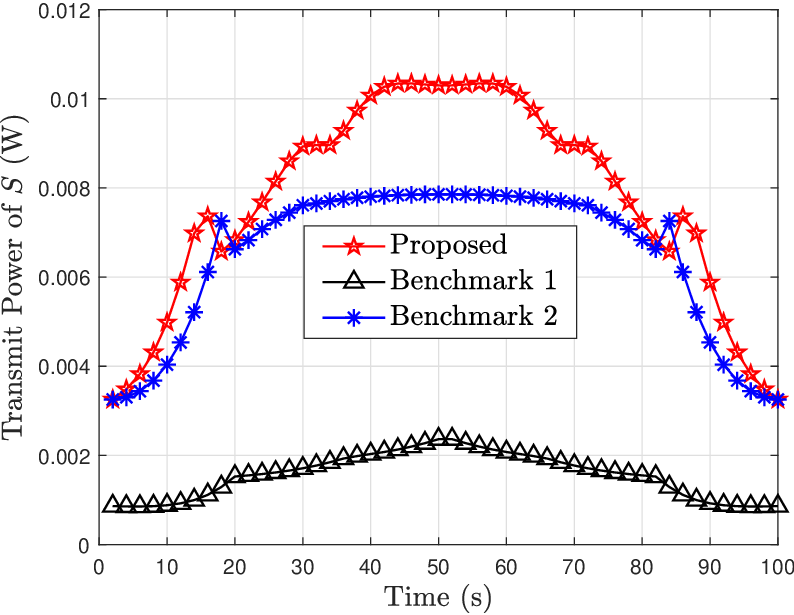}}
	\subfigure[The relationship between the average secrecy rate and the
	number of iterations.]{
		\label{fig05d}
		\includegraphics[width = 0.4  \textwidth]{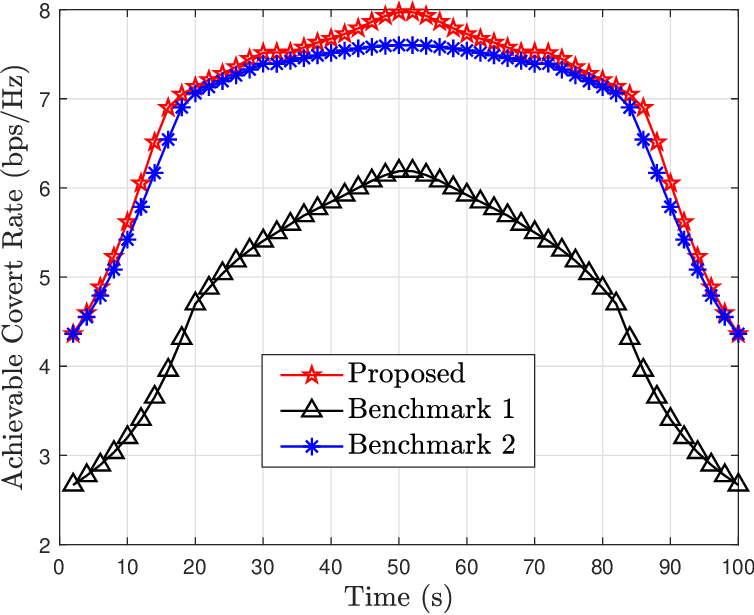}}
	\caption{Scenario 2: The covert systems with $N_W = 3$ wardens with $K =6$ antennas.}
	\label{fig05m}
\end{figure}
Figs. \ref{fig04m} and \ref{fig05m} illustrate the optimized trajectory, transmit power and achievable covert rate for the different schemes, respectively. 
For Benchmark 1 where $S$ flies with a given trajectory, the transmit power should be reduced as it approaches the wardens' area to avoid the information being detected. 
While for Benchmark 2 where $J$ flies with a given trajectory, it can be found that when $J$'s trajectory deviates from the optimal trajectory and moves away from the warden that should be close, $S$'s trajectory will also move away from the warden accordingly to prevent the information being detected. 
In the proposed scheme, $S$ is as close as possible to legitimate users and increases the transmit power while satisfying the quality of covert communication to maximize the achievable covert rate. 
At the same time, to make $S$ can be as close as possible to the legitimate user and increase the transmit power, $J$ will get as close as possible to the most threatening warden under the current time slot. 
In the proposed scheme, the characteristics of benchmark 1 and benchmark 2 are considered simultaneously to obtain optimal covert performance.

\begin{figure}[t]
	\centering
	\subfigure[Average covert rate versus varying $I$.]{
		\label{fig06a}
		\includegraphics[width = 0.4  \textwidth]{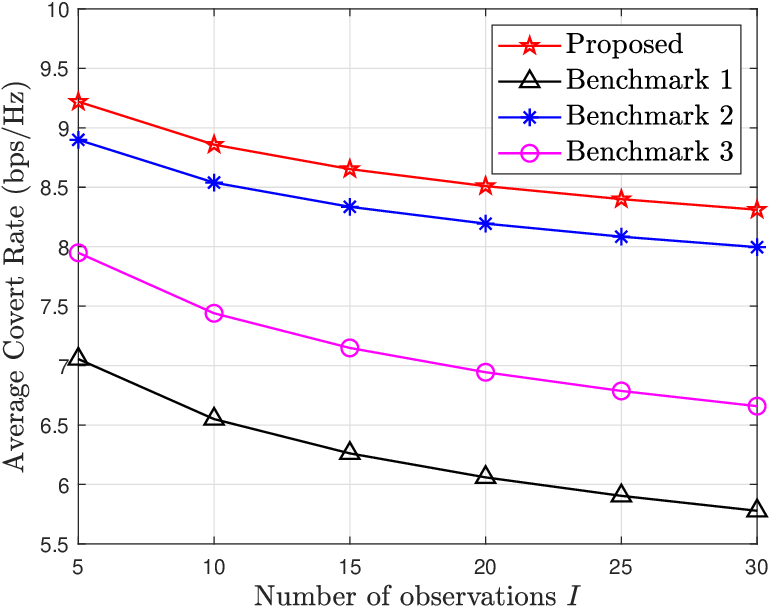}}
	\subfigure[Average transmit power of $S$ versus varying $I$.]{
		\label{fig06b}
		\includegraphics[width = 0.4  \textwidth]{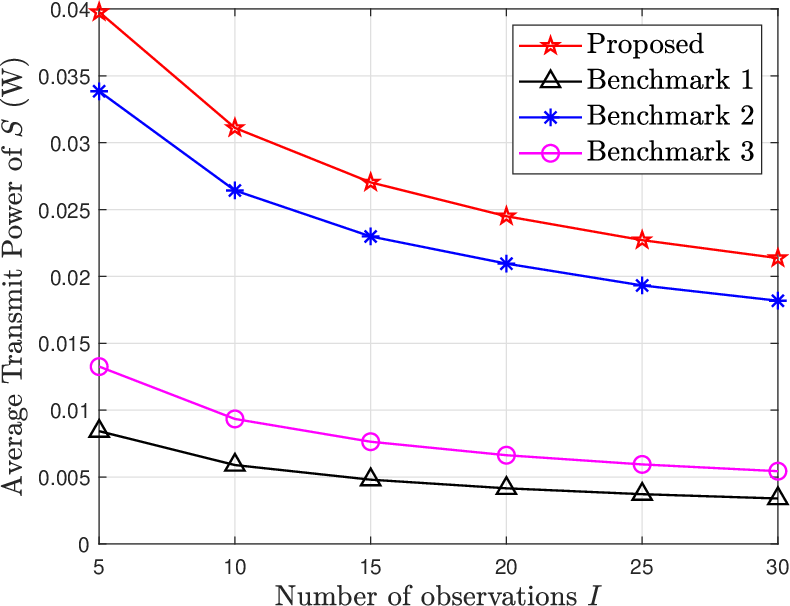}}
	\caption{The impact of varying number of observations $I$.}
	\label{fig06}
\end{figure}
Fig. \ref{fig06} compares the performance between the proposed scheme and other schemes with varying observations number $I$. 
One can see from Fig. \ref{fig06} that the average covert rate achieved by all schemes decreases significantly as $I$ increases. 
The observation number $ I $ considerably impacts the average covert rate of the considered system. 
The reason is that when $I$ increases, wardens' monitor capability becomes more powerful, enabling the $S$ to transmit lower power than before. 
Meanwhile, the average covert rate decreases at a slower rate as $I$ becomes larger can be observed, which reflects when $I$ is large enough, increasing $I$ again will have less gain on the monitor ability of wardens.

\begin{figure}[t]
	\centering
	\subfigure[Average covert rate versus varying $\varepsilon $.]{
		\label{fig07a}
		\includegraphics[width = 0.4  \textwidth]{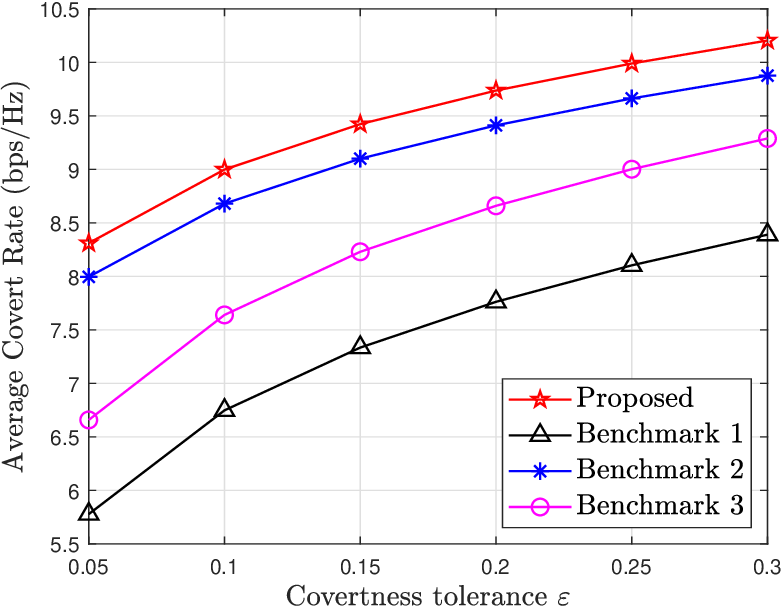}}
	\subfigure[Average transmit power of $S$ versus varying $\varepsilon $.]{
		\label{fig07b}
		\includegraphics[width = 0.4  \textwidth]{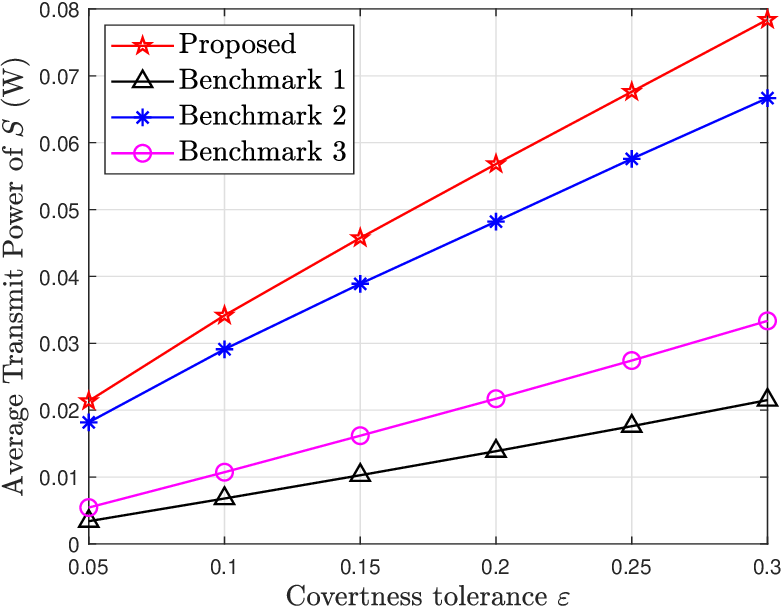}}
	\caption{The impact of varying covertness tolerance $\varepsilon $.}
	\label{fig07}
\end{figure}
Fig. \ref{fig07} plots the average covert rate of the system with covertness tolerance $\varepsilon $ for different schemes. 
One can observe from Fig. \ref{fig07} that the average covert rate achieved by all schemes increases significantly as $\varepsilon $ increases. 
This is because when $\varepsilon $ increases, requirements for covertness become easier to meet, enabling the $S$ to transmit higher power. 
Compared with given $J$'s trajectory, optimizing the hovering position of $J$ can effectively cope with the monitoring of wardens. 
This is because the most threatening warden at different time slots may be different.

\begin{figure}[t]
	\centering
	\subfigure[Average covert rate versus varying transmit power of $J$.]{
		\label{fig08a}
		\includegraphics[width = 0.4  \textwidth]{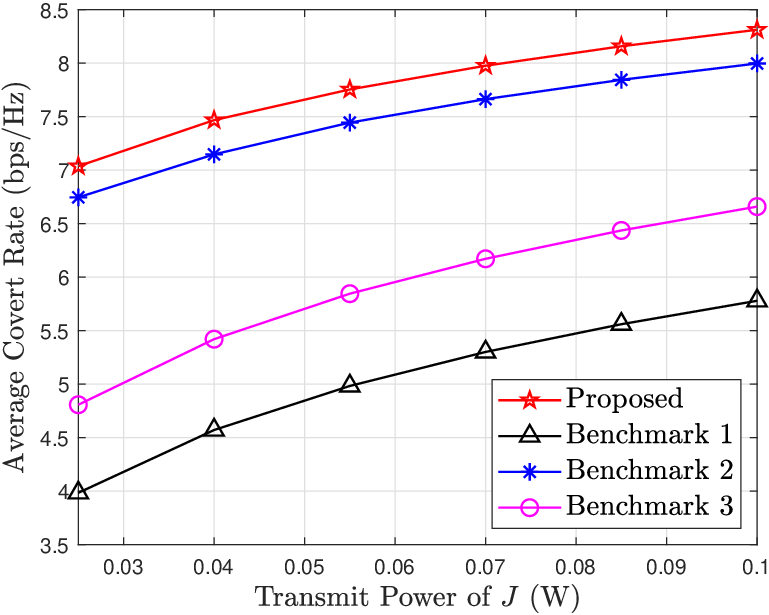}}
	\subfigure[Average transmit power of $S$ versus varying transmit power of $J$.]{
		\label{fig08b}
		\includegraphics[width = 0.4  \textwidth]{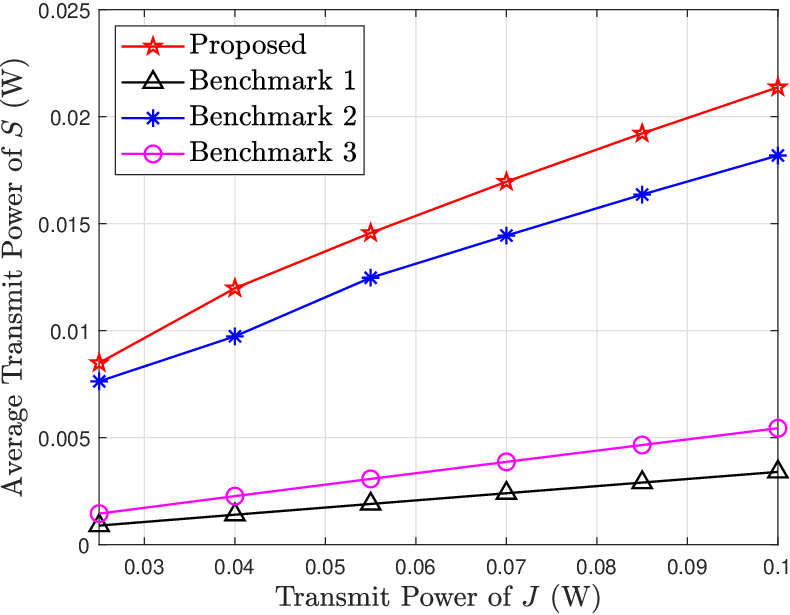}}
	\caption{The impact of varying transmit power of $J$.}
	\label{fig08}
\end{figure}
Fig. \ref{fig08} illustrates the effect of the transmit power of $J$ on the average covert rate with different schemes. 
It can be observed that as the transmit power of $J$ increases, the curves corresponding to all the schemes show an increasing trend. 
This is because the higher power jamming signals can more effectively confuse the detection wardens, making it more difficult for them to detect the presence or absence of covert information from the jamming. 
It is also worth noting that the gain in average covert rate from increasing transmit power of $J$ is diminishing, which shows that simply increasing the transmitting power of $ J $ can not enhance the covert performance always. 

\begin{figure}[t]
	\centering
	\subfigure[Average covert rate versus varying radius of uncertain region of wardens.]{
		\label{fig09a}
		\includegraphics[width = 0.4  \textwidth]{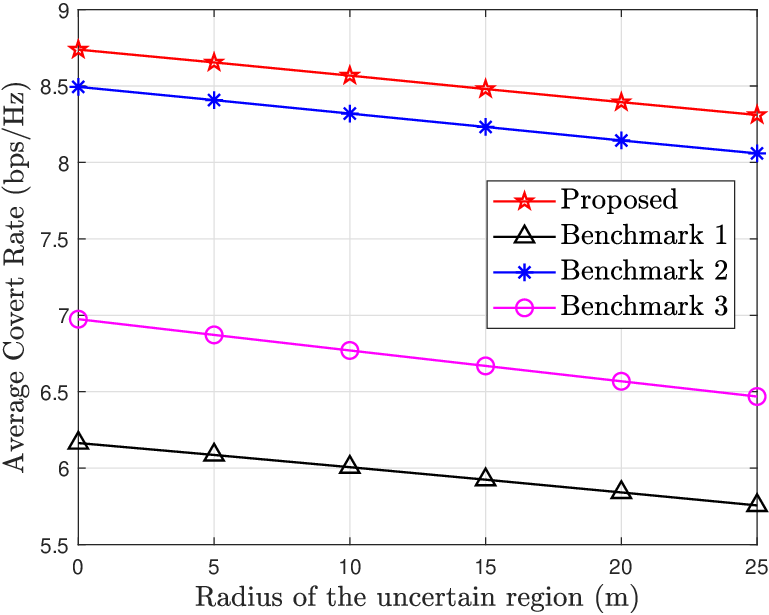}}
	\subfigure[Average transmit power of $S$ versus varying radius of uncertain region of wardens.]{
		\label{fig09b}
		\includegraphics[width = 0.4  \textwidth]{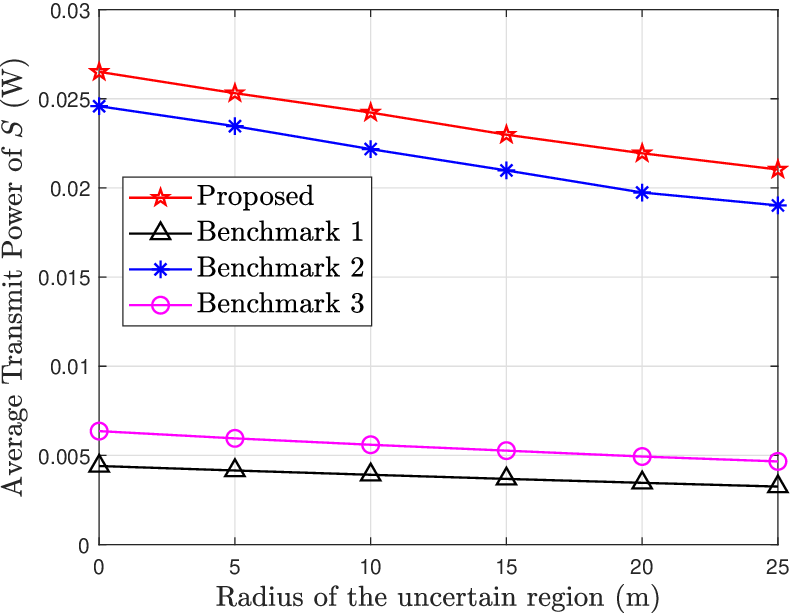}}
	\caption{The impact of varying radius of uncertain region of wardens.}
	\label{fig09}
\end{figure}
Fig. \ref{fig09} shows the relationship between the radius of the uncertainty region and the average covert rate.
One can observe that the average covert rate has a decreasing trend with increasing the uncertainty region. 
The reason is that when the uncertain region radius increases, $S$ obtaining the perfect location knowledge will become more challenging. 
At this point, UAVs had to choose a more conservative transmission method to ensure that the covert information is not detected. 
The uncertainty area radius of wardens, which reflects the channel state information between the UAV and the monitoring node, considerably impacts the average covert rate of the considered system. 
Moreover, we can also see from Fig. \ref{fig06} - \ref{fig09} that the covert performance of the proposed scheme outperforms that of others. 
In contrast, Benchmark 1 wherein $S$ works in the fixed trajectory has the lowest covert performance.


\begin{figure}[t]
	\centering
	\includegraphics[width = 3in]{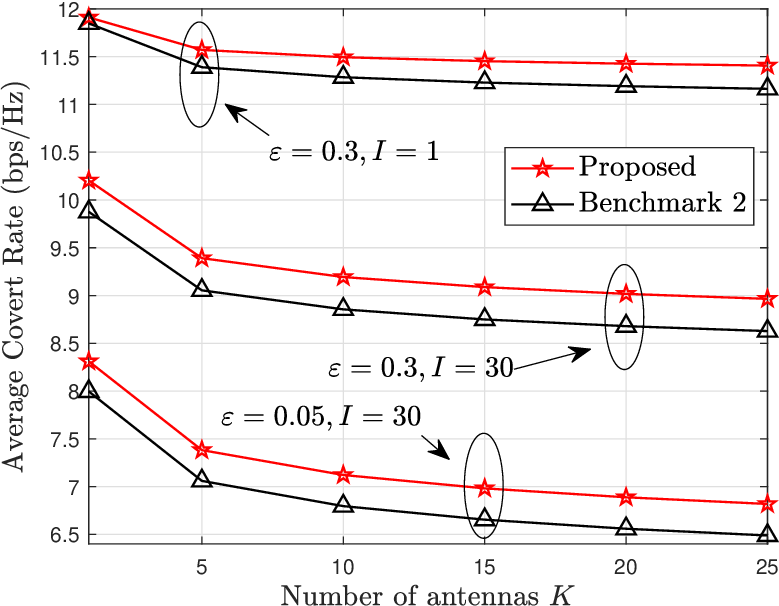}
	\caption{Average covert rate versus number of antennas of wardens.}
	\label{fig10}
\end{figure}
The effect of the number of antennas the wardens are equipped with on the average covert rate is shown in Fig. \ref{fig10}. 
It shows that the average covert rate for all the schemes decreases when the number of antennas equipped by the wardens increases. 
This is because, with the increase in antennas, equipped wardens can receive more information and, to a certain extent, eliminate noise interference. 
In addition, by looking at the curves in Fig. \ref{fig10}, we can also find that the larger the $I$ is, the steeper the downward trend of the curve in the same scheme. 
Similarly, the scheme's smaller $\varepsilon $ corresponds to a steeper downward curve trend. 
This is because for larger $I$, the gain from each increased antenna will also increase, and for $\varepsilon $ is similar to this.

\section{Conclusion}
\label{sec:Conclusion}

This work investigated a dual collaborative UAV system for covert communication in a multi-user scenario on the uncertainty of multiple non-colluding monitoring nodes locations. 
The scenario with multiple single-antenna and multiple-antenna wardens was considered, and the covertness constraint was simplified by analyzing 
the DEP with uncertain locations of wardens.
By jointly designing the trajectory and transmit power of the aerial base station and the trajectory of the aerial jammer, maximizing the minimum average covert rate was formulated as a multivariate coupled non-convex maximization problem. 
A new efficient algorithm was proposed to solve this challenging problem based on the SCA technique, and the KKT solutions were obtained. 
Numerical results show the proposed algorithm's efficiency and each parameter's impact on the average covert rate. 


\begin{appendices}	
	\section{Proof of Lemma 2}
	\label{appendicesA}

By the previous description and the definition of KL-divergence \cite{Lehmann2006Book}, the following equation is obtained
\begin{equation}
	\begin{aligned}
		\mathbb{D}\left( {P_0^I||P_1^I} \right) &= I\mathbb{D}\left( {{P_0}||{P_1}} \right) \\
		&= 2I\int_{{\mathbf{x}} \in {\mathbb{R}^K}} {{f_0}\left( {\mathbf{x}} \right)\ln \left( {\frac{{{f_0}\left( {\mathbf{x}} \right)}}{{{f_1}\left( {\mathbf{x}} \right)}}} \right)d{\mathbf{x}}},
		\label{KLdivergence1}
	\end{aligned}
\end{equation}
where
${f_0}\left( {\bf{x}} \right) = \frac{{\exp \left( { - \frac{1}{2}{{\bf{x}}^T}{\bf{D}}_0^{ - 1}\left( n \right){\bf{x}}} \right)}}{{{{\left( {2\pi } \right)}^{\frac{K}{2}}}\det {{\left( {{{\bf{D}}_0}\left( n \right)} \right)}^{\frac{1}{2}}}}}$, 
${f_1}\left( {\bf{x}} \right) = \frac{{\exp \left( { - \frac{1}{2}{{\bf{x}}^T}{\bf{D}}_1^{ - 1}\left( n \right){\bf{x}}} \right)}}{{{{\left( {2\pi } \right)}^{\frac{K}{2}}}\det {{\left( {{{\bf{D}}_1}\left( n \right)} \right)}^{\frac{1}{2}}}}}$, 
${{\mathbf{D}}_0}\left( n \right) = \frac{{{{\mathbf{K}}_0}\left( n \right)}}{2}$, 
and 
${{\mathbf{D}}_1}\left( n \right) = \frac{{{{\mathbf{K}}_1}\left( n \right)}}{2}$. 
With some simple algebraic manipulations, we obtain 
	\begin{equation}
		\begin{aligned}
			\mathbb{D}\left( {P_0^I||P_1^I} \right) =& \underbrace {\int_{{\bf{x}} \in {\mathbb{R}^K}} {\frac{{\exp \left( { - \frac{1}{2}{{\bf{x}}^T}{\bf{D}}_0^{ - 1}\left( n \right){\bf{x}}} \right)}}{{{{\left( {2\pi } \right)}^{\frac{K}{2}}}\det {{\left( {{{\bf{D}}_0}\left( n \right)} \right)}^{\frac{1}{2}}}}}\ln \left( {\frac{{\det \left( {{{\bf{D}}_1}\left( n \right)} \right)}}{{\det \left( {{{\bf{D}}_0}\left( n \right)} \right)}}} \right)d{\bf{x}}} }_{ \buildrel \Delta \over = {D_1}}\\
			& \underbrace { - \int_{{\bf{x}} \in {\mathbb{R}^K}} {\frac{{{{\bf{x}}^T}{\bf{D}}_0^{ - 1}\left( n \right){\bf{x}}}}{{{{\left( {2\pi } \right)}^{\frac{K}{2}}}\det {{\left( {{{\bf{D}}_0}\left( n \right)} \right)}^{\frac{1}{2}}}}}\exp \left( { - \frac{1}{2}{{\bf{x}}^T}{\bf{D}}_0^{ - 1}\left( n \right){\bf{x}}} \right)d{\bf{x}}} }_{ \buildrel \Delta \over = {D_2}}\\
			&+ \underbrace {\int_{{\bf{x}} \in {\mathbb{R}^K}} {\frac{{{{\bf{x}}^T}{\bf{D}}_1^{ - 1}\left( n \right){\bf{x}}}}{{{{\left( {2\pi } \right)}^{\frac{K}{2}}}\det {{\left( {{{\bf{D}}_0}\left( n \right)} \right)}^{\frac{1}{2}}}}}\exp \left( { - \frac{1}{2}{{\bf{x}}^T}{\bf{D}}_0^{ - 1}\left( n \right){\bf{x}}} \right)d{\bf{x}}} }_{ \buildrel \Delta \over = {D_3}}.
			\label{KLdivergence2}
		\end{aligned}
	\end{equation}

\setcounter{equation}{62} 
Since ${\mathbf{D}}_0^{ - 1}\left( n \right)$ is a symmetric matrix, then we have 
${\bf{D}}_0^{ - 1}\left( n \right) = {\bf{Q}}_0^T\left( n \right)\theta _0^{ - 1}\left( n \right){{\bf{Q}}_0}\left( n \right)$, 
where 
${{\bf{Q}}_0}\left( n \right)$ is an orthogonal matrix. 
Then we have 
$$\exp \left( { - \frac{1}{2}{{\bf{x}}^T}{\bf{D}}_0^{ - 1}\left( n \right){\bf{x}}} \right) = \exp \left( { - \frac{1}{2}{{\bf{x}}^T}{\bf{Q}}_0^T\left( n \right)\theta _0^{ - 1}\left( n \right){{\bf{Q}}_0}\left( n \right){\bf{x}}} \right).$$
Denote ${\bf{y}} = {{\bf{Q}}_0}\left( n \right){\bf{x}}$, we have 
${\bf{x}} = {\bf{Q}}_0^T\left( n \right){\bf{y}}$ 
and 
$d{\bf{x}} = \det \left( {{\bf{Q}}_0^T\left( n \right)} \right)d{\bf{y}}$, respectively.
Thus, with variable substitution, 
${D_1}$ is written in the following form
\begin{equation}
	\begin{aligned}
		{D_1} =& \int_{{\bf{y}} \in {\mathbb{R}^K}} {\frac{{\exp \left( { - \frac{1}{2}{{\bf{y}}^T}\theta _0^{ - 1}\left( n \right){\bf{y}}} \right)}}{{{{\left( {2\pi } \right)}^{\frac{K}{2}}}\tau _1^{\frac{1}{2}} \cdots \tau _K^{\frac{1}{2}}}}}  \ln \left( {\frac{{\det \left( {{{\bf{D}}_1}\left( n \right)} \right)}}{{\det \left( {{{\bf{D}}_0}\left( n \right)} \right)}}} \right)\det \left( {{\bf{Q}}_0^T\left( n \right)} \right)d{\bf{y}}\\
		= &\ln \left( {\frac{{\det \left( {{{\mathbf{D}}_1}\left( n \right)} \right)}}{{\det \left( {{{\mathbf{D}}_0}\left( n \right)} \right)}}} \right) \times \int_{ - \infty }^{ + \infty } {\frac{1}{{{{\left( {2\pi } \right)}^{\frac{1}{2}}}\tau _1^{\frac{1}{2}}}}\exp \left( { - \frac{1}{2}\tau _1^{ - 1}y_1^2} \right)d{y_1}}  \times \cdots \\
		&\times \int_{ - \infty }^{ + \infty } {\frac{1}{{{{\left( {2\pi } \right)}^{\frac{1}{2}}}\tau _K^{\frac{1}{2}}}}\exp \left( { - \frac{1}{2}\tau _K^{ - 1}y_K^2} \right)d{y_K}}\\
		\mathop  = \limits^{\left( a \right)}  &\ln \left( {\frac{{\det \left( {{{\mathbf{D}}_1}\left( n \right)} \right)}}{{\det \left( {{{\mathbf{D}}_0}\left( n \right)} \right)}}} \right),
		\label{D101}
	\end{aligned}
\end{equation}
where 
${\mathbf{D}}_0^{ - 1}\left( n \right) = {\mathbf{Q}}_0^T\left( n \right)\theta _0^{ - 1}\left( n \right){{\mathbf{Q}}_0}\left( n \right)$, ${\mathbf{y}} = {{\mathbf{Q}}_0}\left( n \right){\mathbf{x}}$, ${{\mathbf{Q}}_0}\left( n \right)$ is unitary matrix, 
${\tau _1}, \cdots, {\tau _K}$ is the eigenvalue of ${{{\mathbf{D}}_0}\left( n \right)}$, 
${\theta _0\left( n \right)}$ is the diagonal matrix composed of these eigenvalues \cite{HornRA2012Book}, 
and 
step ${\left( a \right)}$ is obtained 
due to $\int_{ - \infty }^{ + \infty } {\frac{1}{{{{\left( {2\pi } \right)}^{\frac{1}{2}}}\tau _i^{\frac{1}{2}}}}\exp \left( { - \frac{1}{2}\tau _i^{ - 1}y_i^2} \right)d{y_i}} = 1$.
Based on ${{\mathbf{D}}_0}\left( n \right) = \frac{{{{\mathbf{K}}_0}\left( n \right)}}{2}$ and ${{\mathbf{D}}_1}\left( n \right) = \frac{{{{\mathbf{K}}_1}\left( n \right)}}{2}$, we obtain 
	\begin{subequations}
	\begin{align}
		&\det \left( {{{\mathbf{D}}_0}\left( n \right)} \right) = \frac{{\left( {{P_J}{{\left\| {{{\mathbf{h}}_{J{W_m}}}\left( n \right)} \right\|}^2} + {\sigma ^2}} \right){\sigma ^{2\left( {K - 1} \right)}}}}{{{2^K}}} \label{D0D101}\\ 
		&\det \left( {{{\mathbf{D}}_1}\left( n \right)} \right) = \frac{{\left( {{P_J}{{\left\| {{{\mathbf{h}}_{J{W_m}}}\left( n \right)} \right\|}^2} + {P_S}\left( n \right){{\left\| {{{\mathbf{h}}_{S{W_m}}}\left( n \right)} \right\|}^2} + {\sigma ^2}} \right){\sigma ^{2\left( {K - 1} \right)}}}}{{{2^K}}}. \label{D0D102}
	\end{align}
\end{subequations}
Finally, ${D_1}$ is obtained as 
\begin{equation}
	{D_1} = \ln \left( {\frac{{{P_J}{{\left\| {{{\bf{h}}_{J{W_m}}}\left( n \right)} \right\|}^2} + {P_S}\left( n \right){{\left\| {{{\bf{h}}_{S{W_m}}}\left( n \right)} \right\|}^2} + {\sigma ^2}}}{{{P_J}{{\left\| {{{\bf{h}}_{J{W_m}}}\left( n \right)} \right\|}^2} + {\sigma ^2}}}} \right).
	\label{D103}
\end{equation}

Utilizing a similar method, ${D_2}$ is rewritten as
\begin{equation}
	\begin{aligned}
		{D_2} = & - \int_{{\bf{y}} \in {\mathbb{R}^K}} {\frac{{{{\bf{y}}^T}\theta _0^{ - 1}\left( n \right){\bf{y}}\exp \left( { - \frac{1}{2}{{\bf{y}}^T}\theta _0^{ - 1}\left( n \right){\bf{y}}} \right)\det \left( {{\bf{Q}}_0^T\left( n \right)} \right)d{\bf{y}}}}{{{{\left( {2\pi } \right)}^{\frac{K}{2}}}\tau _1^{\frac{1}{2}} \cdots \tau _K^{\frac{1}{2}}}}} \\
		= & - \int_{{\bf{y}} \in {\mathbb{R}^K}} {\frac{{\tau _1^{ - 1}y_1^2}}{{{{\left( {2\pi } \right)}^{\frac{K}{2}}}\tau _1^{\frac{1}{2}} \cdots \tau _K^{\frac{1}{2}}}}\exp \left( { - \frac{1}{2}{{\bf{y}}^T}\theta _0^{ - 1}\left( n \right){\bf{y}}} \right)d{\bf{y}}} -  \cdots \\
		&- \int_{{\bf{y}} \in {\mathbb{R}^K}} {\frac{{\tau _K^{ - 1}y_K^2}}{{{{\left( {2\pi } \right)}^{\frac{K}{2}}}\tau _1^{\frac{1}{2}} \cdots \tau _K^{\frac{1}{2}}}}\exp \left( { - \frac{1}{2}{{\bf{y}}^T}\theta _0^{ - 1}\left( n \right){\bf{y}}} \right)d{\bf{y}}}\\
		=& - \sum\limits_{i = 1}^K {\int_{{\bf{y}} \in {\mathbb{R}^K}} {\frac{{\tau _i^{ - 1}y_i^2}}{{{{\left( {2\pi } \right)}^{\frac{K}{2}}}\tau _1^{\frac{1}{2}} \cdots \tau _K^{\frac{1}{2}}}}\exp \left( { - \frac{1}{2}{{\bf{y}}^T}\theta _0^{ - 1}\left( n \right){\bf{y}}} \right)d{\bf{y}}} }.
		\label{D200}
	\end{aligned}
\end{equation}
Due to 
\begin{equation}
	\begin{aligned}
		&\int_{{\bf{y}} \in {\mathbb{R}^K}} {\frac{{\tau _i^{ - 1}y_i^2}}{{{{\left( {2\pi } \right)}^{\frac{K}{2}}}\tau _1^{\frac{1}{2}} \cdots \tau _K^{\frac{1}{2}}}}\exp \left( { - \frac{1}{2}{{\bf{y}}^T}\theta _0^{ - 1}\left( n \right){\bf{y}}} \right)d{\bf{y}}} \\
		=& \int_{ - \infty }^{ + \infty } {\frac{{\tau _i^{ - 1}y_i^2}}{{{{\left( {2\pi } \right)}^{\frac{1}{2}}}\tau _1^{\frac{1}{2}}}}\exp \left( { - \frac{1}{2}\tau _1^{ - 1}y_1^2} \right)d{y_1}}  \times \cdots \times \int_{ - \infty }^{ + \infty } {\frac{1}{{{{\left( {2\pi } \right)}^{\frac{1}{2}}}\tau _K^{\frac{1}{2}}}}\exp \left( { - \frac{1}{2}\tau _K^{ - 1}y_K^2} \right)d{y_K}} \\
		= &1,
		\label{liancheng1}
	\end{aligned}
\end{equation}
we obtain ${D_2} = - K$.

Based on the defintion ${{\bf{K}}_1}\left( n \right)$ and ${{\bf{K}}_1}\left( n \right)$, we obtain
\begin{equation}
	\begin{aligned}
		{{\mathbf{D}}_0}\left( n \right) = &\frac{{{\sigma ^2}{\mathbf{I}} + {P_J}{{\mathbf{h}}_{J{W_m}}}\left( n \right){\mathbf{h}}_{J{W_m}}^H\left( n \right)}}{2} \\
		= & \frac{{{\sigma ^2}{\mathbf{I}} + {P_J}{{\left| {{h_{J{W_m}}}\left( n \right)} \right|}^2}{\mathbf{1}}}}{2}
		\label{D_0(n)}
	\end{aligned}
\end{equation}
and
\begin{equation}
	\begin{aligned}
		{{\mathbf{D}}_1}\left( n \right) &= \frac{{{\sigma ^2}{\mathbf{I}} + {{\mathbf{g}}_m}\left( n \right){\mathbf{g}}_m^H\left( n \right)}}{2} \\
		&= \frac{{{\sigma ^2}{\mathbf{I}} + \left( {{P_J}{{\left| {{h_{J{W_m}}}\left( n \right)} \right|}^2} + {P_S}\left( n \right){{\left| {{h_{S{W_m}}}\left( n \right)} \right|}^2}} \right){\mathbf{1}}}}{2},
		\label{D_1(n)}
	\end{aligned}
\end{equation}
respectively, 
where ${\mathbf{1}}$ is $K \times K$ matrix of unit elements and ${\mathbf{I}}$ is $K \times K$ identity matrix. 
Then, similar to ${\mathbf{D}}_0^{ - 1}\left( n \right)$, we obtian 
${\mathbf{D}}_1^{ - 1}\left( n \right) = {\mathbf{Q}}_0^T\left( n \right)\theta _1^{ - 1}\left( n \right){{\mathbf{Q}}_0}\left( n \right)$, where $\theta _1\left( n \right)$ is the diagonal matrix composed of ${\mathbf{D}}_1\left( n \right)$'s eigenvalues. 
Thus, ${D_3}$ is written as
\begin{equation}
	\begin{aligned}
	  {D_3} =& \int_{{\bf{y}} \in {\mathbb{R}^K}} {\frac{{{{\bf{y}}^T}\theta _1^{ - 1}\left( n \right){\bf{y}}\exp \left( { - \frac{1}{2}{{\bf{y}}^T}\theta _0^{ - 1}\left( n \right){\bf{y}}} \right)\det \left( {{\bf{Q}}_0^T\left( n \right)} \right)d{\bf{y}}}}{{{{\left( {2\pi } \right)}^{\frac{K}{2}}}\tau _1^{\frac{1}{2}} \cdots \tau _K^{\frac{1}{2}}}}} \\
	  = & \int_{{\bf{y}} \in {\mathbb{R}^K}} {\frac{{\kappa _1^{ - 1}y_1^2}}{{{{\left( {2\pi } \right)}^{\frac{K}{2}}}\tau _1^{\frac{1}{2}} \cdots \tau _K^{\frac{1}{2}}}}\exp \left( { - \frac{1}{2}{{\bf{y}}^T}\theta _0^{ - 1}\left( n \right){\bf{y}}} \right)d{\bf{y}}}  +  \cdots \\
	  &+ \int_{{\bf{y}} \in {\mathbb{R}^K}} {\frac{{\kappa _K^{ - 1}y_K^2}}{{{{\left( {2\pi } \right)}^{\frac{K}{2}}}\tau _1^{\frac{1}{2}} \cdots \tau _K^{\frac{1}{2}}}}\exp \left( { - \frac{1}{2}{{\bf{y}}^T}\theta _0^{ - 1}\left( n \right){\bf{y}}} \right)d{\bf{y}}} \\
	  = & \sum\limits_{j = 1}^K {\int_{{\bf{y}} \in {\mathbb{R}^K}} {\frac{{\kappa _j^{ - 1}y_j^2}}{{{{\left( {2\pi } \right)}^{\frac{K}{2}}}\tau _1^{\frac{1}{2}} \cdots \tau _K^{\frac{1}{2}}}}\exp \left( { - \frac{1}{2}{{\bf{y}}^T}\theta _0^{ - 1}\left( n \right){\bf{y}}} \right)d{\bf{y}}} },
	\end{aligned}
\end{equation}
where ${\kappa _1}, \cdots, {\kappa _K}$ is the eigenvalues of ${\mathbf{D}}_1\left( n \right)$. 
Similar to (\ref{liancheng1}), we have 
\begin{equation}
	\begin{aligned}
		&\int_{{\bf{y}} \in {\mathbb{R}^K}} {\frac{{\kappa _j^{ - 1}y_j^2}}{{{{\left( {2\pi } \right)}^{\frac{K}{2}}}\tau _1^{\frac{1}{2}} \cdots \tau _K^{\frac{1}{2}}}}\exp \left( { - \frac{1}{2}{{\bf{y}}^T}\theta _0^{ - 1}\left( n \right){\bf{y}}} \right)d{\bf{y}}} \\
		=& \int_{ - \infty }^{ + \infty } {\frac{{\kappa _j^{ - 1}y_j^2}}{{{{\left( {2\pi } \right)}^{\frac{1}{2}}}\tau _1^{\frac{1}{2}}}}\exp \left( { - \frac{1}{2}\tau _1^{ - 1}y_1^2} \right)d{y_1}}  \times  \cdots \times \int_{ - \infty }^{ + \infty } {\frac{1}{{{{\left( {2\pi } \right)}^{\frac{1}{2}}}\tau _K^{\frac{1}{2}}}}\exp \left( { - \frac{1}{2}\tau _K^{ - 1}y_K^2} \right)d{y_K}} \\
		=& \frac{{{\tau _j}}}{{{\kappa _j}}}. 
		\label{liancheng2}
	\end{aligned}
\end{equation}
Thus, we obtain ${D_3} = \frac{{{\tau _1}}}{{{\kappa _1}}} +  \cdots  + \frac{{{\tau _K}}}{{{\kappa _K}}}$. 
According to (\textrm{\ref{D_0(n)}}), (\textrm{\ref{D_1(n)}}), due to the eigenvalues of ${\mathbf{1}}$ is $K,0, \cdots ,0$, we obtain the eigenvalues of ${{\mathbf{D}}_0}\left( n \right)$ and ${{\mathbf{D}}_1}\left( n \right)$ as 
$\frac{{{P_J}K{{\left| {{h_{J{W_m}}}\left( n \right)} \right|}^2} + {\sigma ^2}}}{2},\frac{{{\sigma ^2}}}{2}$, $\cdots$, $\frac{{{\sigma ^2}}}{2}$ 
and 
$\frac{{{P_J}K{{\left| {{h_{J{W_m}}}\left( n \right)} \right|}^2} + {P_S}\left( n \right)K{{\left| {{h_{S{W_m}}}\left( n \right)} \right|}^2} + {\sigma ^2}}}{2},\frac{{{\sigma ^2}}}{2}$, $\cdots,$ $\frac{{{\sigma ^2}}}{2}$, respectively. 
Finally, ${D_3}$ is obtained as 
\begin{equation}
	\begin{aligned}
		{D_3} = &\frac{{{P_J}{{\left\| {{{\bf{h}}_{J{W_m}}}\left( n \right)} \right\|}^2} + {\sigma ^2}}}{{{P_J}{{\left\| {{{\bf{h}}_{J{W_m}}}\left( n \right)} \right\|}^2} + {P_S}\left( n \right){{\left\| {{{\bf{h}}_{S{W_m}}}\left( n \right)} \right\|}^2} + {\sigma ^2}}}  + K - 1.
		\label{H232}
	\end{aligned}
\end{equation}

Then, we obtain 
	\begin{equation}
		\begin{aligned}
			\mathbb{D}\left( {P_0^I||P_1^I} \right)  = & \ln \left( {\frac{{{\sigma ^2} + {P_J}{{\left\| {{{\mathbf{h}}_{J{W_m}}}\left( n \right)} \right\|}^2} + {P_S}\left( n \right){{\left\| {{{\mathbf{h}}_{S{W_m}}}\left( n \right)} \right\|}^2}}}{{{\sigma ^2} + {P_J}{{\left\| {{{\mathbf{h}}_{J{W_m}}}\left( n \right)} \right\|}^2}}}} \right) \\
			&+ \frac{{{\sigma ^2} + {P_J}{{\left\| {{{\mathbf{h}}_{J{W_m}}}\left( n \right)} \right\|}^2}}}{{{\sigma ^2} + {P_J}{{\left\| {{{\mathbf{h}}_{J{W_m}}}\left( n \right)} \right\|}^2} + {P_S}\left( n \right){{\left\| {{{\mathbf{h}}_{S{W_m}}}\left( n \right)} \right\|}^2}}} - 1\\
			=& I\left( {\ln \left( {1 + \frac{{{P_S}\left( n \right){{\left\| {{{\mathbf{h}}_{S{W_m}}}\left( n \right)} \right\|}^2}}}{{{\sigma ^2} + {P_J}{{\left\| {{{\mathbf{h}}_{J{W_m}}}\left( n \right)} \right\|}^2}}}} \right) - \frac{{\frac{{{P_S}\left( n \right){{\left\| {{{\mathbf{h}}_{S{W_m}}}\left( n \right)} \right\|}^2}}}{{{\sigma ^2} + {P_J}{{\left\| {{{\mathbf{h}}_{J{W_m}}}\left( n \right)} \right\|}^2}}}}}{{1 + \frac{{{P_S}\left( n \right){{\left\| {{{\mathbf{h}}_{S{W_m}}}\left( n \right)} \right\|}^2}}}{{{\sigma ^2} + {P_J}{{\left\| {{{\mathbf{h}}_{J{W_m}}}\left( n \right)} \right\|}^2}}}}}} \right) \\
			= & I\left( {\ln \left( {1 + {\gamma _{m,2}}\left( n \right)} \right) - \frac{{{\gamma _{m,2}}\left( n \right)}}{{1 + {\gamma _{m,2}}\left( n \right)}}} \right)
			\label{KLdivergence3}
		\end{aligned}
	\end{equation}

The proof of \textbf{Lemma 2} is now completed.

\end{appendices}

\end{document}